\documentclass[11pt]{article}
\usepackage{multirow}
\usepackage{graphicx}
\usepackage{textcomp}
\usepackage{natbib}
\usepackage{amsmath,amsfonts,amssymb,amsthm}
\usepackage{xcolor}
\usepackage{mathtools}
\usepackage{booktabs}
\usepackage{authblk}

\newtheorem{theorem}{Theorem}[section]

\newtheorem{remark}{Remark}[section]

\usepackage[a4paper, total={6in, 8in}]{geometry}
\providecommand{\ie}
{
  I\!E
}
\DeclareMathOperator{\cov}{cov}
\DeclareMathOperator{\var}{var}
\DeclareMathOperator{\cm}{cm}
\DeclareMathOperator*{\E}{\mathbb{E}}

\begin{document}
\bibliographystyle{apalike}

\title{Method of Moments Estimation for Affine Stochastic Volatility Models}
\author[a]{Yan-Feng Wu}
\author[b]{Xiangyu Yang \thanks{Corresponding author:  yangxiangyu@email.sdu.edu.cn }}
\author[a]{Jian-Qiang Hu}
\affil[a]{School of Management, Fudan University,       670 Guoshun Road, Shanghai 200433, China.}
\affil[b]{School of Management, Shandong University,  27 Shanda Nanlu,  Jinan 250100, China.}
\renewcommand*{\Affilfont}{\small\it}
\date{}
\maketitle

\begin{abstract}
\noindent
We develop moment estimators for the parameters of affine stochastic volatility models. We first address the challenge of calculating moments for the models by introducing a recursive equation for deriving closed-form expressions for moments of any order. Consequently, we propose our moment estimators. We then establish a central limit theorem for our estimators and derive the explicit formulas for the asymptotic covariance matrix. Finally, we provide numerical results to validate our method.
\end{abstract}
\emph{keywords}: affine jump diffusion; stochastic volatility; Heston model; method of moments.\\%
MOS subject classification: 62F12; 62M05; 60J25.

\section{Introduction}
Volatility plays a crucial role in the pricing of options and other derivative securities. Empirical evidence from financial
markets, such as volatility clustering, the dependence between increments, and volatility smiles, shows that the assumption of constant volatility used in the Black-Scholes European options pricing model is inappropriate (see \citealp{fama1965behavior}).  Therefore, several 
stochastic volatility (SV) models, such as \citet{heston1993closed}, \citet{bates1996jumps} and \citet{barndorff2001non}, have been proposed to better capture the time-varying features of volatility. However, parameter estimation for these SV models poses a formidable computational challenge for practical implementation. Likelihood-based inference and moment-based inference are two major types of estimation methods. Our work in this paper offers a new and direct method of moments (MM) estimation for affine SV models.

There is a large body of literature on likelihood-based inference, in which Markov chain Monte Carlo (MCMC) and maximum likelihood estimation (MLE) are the two main approaches used. An MCMC model assumes a prior distribution of parameters, and then iteratively collects samples in a Bayesian framework until certain convergence is achieved \citep{eraker2001mcmc,jacquier2002bayesian,eraker2003impact,roberts2004bayesian,griffin2006inference,jasra2011inference}. In some cases, this approach may become computationally prohibitive because of its slow convergence to a Markov chain equilibrium distribution. On the other hand, the MLE approach adopts a frequentist perspective and often requires a closed-form expression of the state transition density function; however, this is hard to obtain in SV models, let alone optimization. Hence, the MLE approach often involves some sort of approximation or simulation methods. Two prevail schemes are simulated MLE (SMLE) (see, e.g., \citet{durham2002numerical}) and quasi-MLE (QMLE) (see, e.g., \citet{ruiz1994quasi}, \citet{harvey1996estimation}, and \citet{Feunou2018risk}). \citet{hurn2014estimating} design an MLE estimate using a particle filter to approximate the likelihood. \citet{ai2007maximum} propose a closed-form approximation scheme based on the extracted volatility from observed option prices. \citet{andersen2015parametric} employ optimization tools such as  penalized nonlinear least squares to develop a new parametric estimation procedure. A gradient-based simulated MLE estimate is proposed in \citet{peng2014gradient} and \citet{peng2016gradient}. Nevertheless, approximations used in these MLE methods require enormous computational effort, and  many assumptions may be difficult to satisfy in practice.

Modern financial markets give rise to numerous real-time decision-making instances (e.g. high-frequency trading) that request computationally fast and less restrictive  estimators. The moment-based method is a case in point but has yet to be recognized fully. A few important contributions include the simulated MM (SMM); see \citet{duffie1993simulated}), efficient MM (EMM); see \citet{Bansal1994} and \citet{gallant1996moments}, and generalized MM (GMM); see  \cite{singleton2001estimation}, \citet{bollerslev2002estimating}, \citet{jiang2002estimation}, and \citet{chacko2003spectral}. SMM simulates sequences from the target models and  estimates parameters by matching simulated data moments with actual data moments numerically. GMM can be applied when we have more sample moments than the number of parameters. EMM combines the efficiency of MLE with the flexibility of GMM, which is a variant of SMM. In general, the major issue of the moment-based method is the so-called statistical  inefficiency in the sense that the higher order moments are used, the greater likelihood of estimation bias occurs.  Fortunately, we can mitigate this problem by using better computational techniques (e.g., see \citet{wu2019moment} and \citet{wu2022moment}). For instance, \citet{yang2020method} recently develop the first closed-form MM parameter estimators for a class of L\'{e}vy-driven Ornstein-Uhlenbeck stochastic volatility models. 

In this paper, we consider affine SV models, and develop explicit MM estimators for the parameters of these models. To this end, we first derive the analytical formulas for the required moments and covariances of the affine SV models and then use them to develop our MM estimators.
The MM estimators can be expressed explicitly in terms of the sample moments and sample covariances of the observed asset prices. 
In fact, we have established a recursive equation for deriving closed-form expressions for moments of any order, while our MM estimators only require those moments of relatively low order.
We also prove that the large-sample behaviors of our MM estimators satisfy the central limit theorem and derive the explicit formulas for the asymptotic covariance matrix. In addition, numerical experiments are provided to test the efficiency of our method.  The advantage of our estimators  is that they are simple and easy to implement in practice and do not require the availability of high-frequency data or option price data.

The rest of this paper is organized as follows. In Section 2, we present the affine SV model and its baseline SV model, then calculate the moments and covariances of the returns of the asset price. In Section 3, we derive our MM estimators and establish the central limit theorem. Some numerical experiments are provided in Section 4. Section 5 concludes this paper. The appendices contain some of the detailed calculations which are omitted in the main body of the paper and also some extensions of the baseline affine SV model.

\section{Affine Stochastic Volatility Model}

We first delineate the general framework of affine SV models. Subsequently, we employ the Heston model, a prototypical example of an affine SV model, to illustrate the derivation of closed-form expressions for moments and covariances of any order. We then extend this moment derivation methodology to a broader array of affine SV models, encompassing those with jump components in either the price or variance processes, as well as models incorporating multiple latent volatility factors. These extensions are detailed in the appendices.  

Consider $s(t)$ as the price of a certain asset at time $t$. The general affine SV model explored in this study is governed by the following stochastic differential equations (SDEs):
\begin{align}
\frac{ds(t)}{s(t)} &= \mu(\boldsymbol{v}(t)) dt + \sigma(\boldsymbol{v}(t)) dw^s_t + dz^s(t),\\
d\boldsymbol{v}(t) &= \boldsymbol{\mu_v}(\boldsymbol{v}(t)) dt + \boldsymbol{\sigma_v}(\boldsymbol{v}(t)) d\boldsymbol{w^v}(t) + d\boldsymbol{z^v}(t),
\end{align}
where $\boldsymbol{v}(t)$ denotes a vector of latent variance factors, $w^s(t)$ and $\boldsymbol{w^v}(t)$ are potentially correlated Wiener processes, $z^s(t)$ and $\boldsymbol{z^v}(t)$ are possibly correlated compound Poisson processes. The drift functions \textemdash $\mu(\cdot), \boldsymbol{\mu_v}(\cdot)$ and the variance $\sigma^2(\cdot)$ and the covariance matrix $\boldsymbol{\sigma_v}^T(\cdot)\boldsymbol{\sigma_v}(\cdot)$ \textemdash all have affine dependence on the hidden state vector $\boldsymbol{v}(t)$ \citep{duffie2000transform}. It is assumed that the drifts, diffusions and compound Poisson processes are well-behaved enough to ensure the unique strong solution of the SDEs. 

This general model includes several common SV models as special cases, such as the Heston model, the stochastic volatility with jumps (SVJ) model, the stochastic volatility with independent jumps (SVIJ) model, the stochastic volatility with contemporaneous jumps (SVCJ) model, and even the BNS model. Within the framework of the Heston model \citep{heston1993closed}, the vector $\boldsymbol{v}(t)$ simplifies to a single variance factor $v(t)$, with drift and diffusion functions defined as:
\begin{align*}
\mu(v(t)) &= \mu, & \sigma(v(t)) &= \sqrt{v(t)},\\
\mu_v(v(t)) &= k(\theta - v(t)), & \sigma_v(v(t)) &= \sigma_v \sqrt{v(t)},
\end{align*}
and without any jump components $z^s(t)$ and $z^v(t)$ in both the price and variance equations. The SVJ model, as described by \citep{bates1996jumps}, extends the Heston framework by including jumps in the price process. The SVIJ and SVCJ models, detailed in \citet{eraker2003impact}, further incorporate jumps into the variance process, either independently or contemporaneously with price jumps. An additional variant is the multi-factor Heston model, which characterizes the latent variance as a vector of square-root diffusions, i.e., $\boldsymbol{v}(t) = (v_1(t),v_2(t))^T$, and 
\begin{align*}
\mu(\boldsymbol{v}(t)) &= \mu, & \sigma(\boldsymbol{v}(t)) &= \sqrt{v_1(t)+v_2(t)},\\
v_1(t)&= k_1(\theta_1 - v_1(t)), & \sigma_{v1}(v_1(t)) &= \sigma_{v1} \sqrt{v_1(t)},\\
v_2(t)&= k_2(\theta_2 - v_2(t)), & \sigma_{v2}(v_2(t)) &= \sigma_{v2} \sqrt{v_2(t)}.
\end{align*}

In contrast, the BNS model, proposed by \citet{barndorff2001non} posits a one-dimensional non-Gaussian Ornstein-Uhlenbeck (OU) process for the latent variance, driven solely by a compound Poisson process $z^v(t)$, devoid of a diffusion component. The drift and diffusion functions are defined as: 
\begin{align*}
\mu(v(t)) &= \mu, & \sigma(v(t)) &= \sqrt{v(t)},\\
\mu_v(v(t)) &= -\lambda v(t), & \sigma_v(v(t)) &= 0, 
\end{align*} 
and there is no jump components in the price process.

It is notable that these affine SV models do not possess closed-form solution for their transition density functions. However, for the general affine SV model, closed-form expressions for their conditional Characteristic Functions (CFs) can be derived \citep{duffie2000transform, singleton2001estimation, chacko2003spectral}. \citet{jiang2002estimation} extended this work by developing a closed-form CF for a simplified version of the Heston model.

Our work in this paper aims to extend the analytical tractability of these models to include closed-form formulas for moments and covariances of any (positive integer) order, which can then be used to construct moment estimators for the affine SV model. 

Since the Heston model is the baseline affine SV model, we will focus on its moment calculations which will lay the foundation for the more complex moment computations of other affine SV models. The methodology developed for the Heston model can be generalized and applied to estimate other affine SV models. The detailed moment-derivation procedures for these other affine SV models are provided in the appendices. 

\subsection{Baseline Affine Stochastic Volatility Model}
 We use the following Heston model \citep{heston1993closed} as our baseline affine SV model :
\begin{align}
    \frac{ds(t)}{s(t)} &= \mu dt + \sqrt{v(t)}dw^s(t),\label{eqn:price-process}\\
    dv(t) &= k(\theta - v(t))dt + \sigma_v\sqrt{v(t)}dw^v(t),
    \label{eq.v}
\end{align}
where $v(t)$ is a scalar affine diffusion, and $w^s(t)$ and $w^v(t)$ are two Wiener processes with correlation $\rho$. 
Assume that the initial values $s(0)$ and $v(0)$ are independent of each other and also independent of $w^s(t)$ and $w^v(t)$.  The variance process \eqref{eq.v} is a Cox-Ingersoll-Ross process \citep{cox1985theory}, which is also called square-root diffusion. The parameters $\theta>0$, $k>0$, and $\sigma_v$ determine the long-run mean, the mean reversion velocity, and the volatility of the variance process $v(t)$, respectively, and satisfy the condition $\sigma_v^2 \le 2k\theta$.

Based on It\^{o} formula, the log price process can be written as:
\begin{equation}\label{eqn:log-price}
   d\log s(t) = (\mu-\frac{1}{2}v(t))dt + \sqrt{v(t)}dw^s(t).
\end{equation}
Though modeled as a continuous-time process, the asset price is observed at discrete-time instances. Assume we have observations of $s(t)$ at discrete-time $ih$ ($i=0,1,\cdots,N$), and let $s_i \triangleq s(ih)$. Similarly, let $v_i \triangleq v(ih)$, however, we should note that $v_i$ is not observable.  Finally, we define the return $y_i$ as 
$$
y_i \triangleq \log s_i - \log s_{i-1} .  
$$

We can decompose $w^s(t)$ as
\begin{equation*}
    w^s(t) = \rho w^v(t) + \sqrt{1-\rho^2}w(t),
\end{equation*}
where $w(t)$ is another Wiener process  independent of $w^v(t)$.  For notational simplicity, we define:
\begin{align*}
    I_{s,t} \triangleq & \int_{s}^{t}\sqrt{v(u)}dw^v(u), 
    & I_{s,t}^* \triangleq & \int_{s}^{t}\sqrt{v(u)}dw(u),\\
    \ie_{s,t} \triangleq & \int_{s}^{t}e^{ku}\sqrt{v(u)}dw^v(u),
    & IV_{s,t} \triangleq & \int_{s}^{t}v(u)du,
\end{align*}
and
\begin{align*}
    I_i&\triangleq I_{(i-1)h,ih}, 
    &I_i^*&\triangleq I^*_{(i-1)h,ih},
    &\ie_i&\triangleq \ie_{(i-1)h,ih},
    &IV_i&\triangleq IV_{(i-1)h,ih},\\
    I_{i,t}&\triangleq I_{ih,t},
    &I^*_{i,t}&\triangleq I^*_{ih,t},
    &\ie_{i,t}&\triangleq \ie_{ih,t},
    &IV_{i,t}&\triangleq IV_{ih,t}.
\end{align*}
$IV_{s,t}$, $IV_i$, and $IV_{i,t}$ defined above are also referred as integrated variance (volatility).  We can then express $y_i$ as
\begin{equation}
    y_i = \mu h - \frac{1}{2}IV_i + \rho I_i + \sqrt{1-\rho^2} I_i^*.
\label{eq_yn}
\end{equation}

Based on (\ref{eq.v}), the variance process $v(t)$ can be re-written as:
\begin{equation}\label{eqn:v(t)}
    v(t) = e^{-k(t-s)}v(s)  +\theta \left[1-e^{-k(t-s)}\right] + \sigma_v e^{-kt}\ie_{s,t},
\end{equation}
which is a Markov  process that has a steady-state gamma distribution with mean $\theta$ and variance $\theta \sigma_v^2/(2k)$, e.g., see \cite{cox1985theory}.  Without loss of generality, throughout this paper we assume that $v(0)$ is distributed according to the steady-state distribution of $v(t)$, which implies that $v(t)$ is strictly stationary and ergodic \citep{overbeck1997}.
However, all the results we shall derive in the rest of this paper, including the formulas for all the moments/covariances and the MM estimators, remain the same for any non-negative $v(0)$.   
Since $v(t)$ is stationary, we have
\begin{eqnarray}
    \E[v(t)^m] &=& \prod_{j=0}^{m-1}\left(\theta + \frac{j\sigma_v^2}{2k}\right), \label{eqn:v_moment} \\
    \var (v(t)) &=& \frac{\theta\sigma_v^2}{2k}, 
\end{eqnarray}
for $m=1,2,\ldots$ and $t>0$, and
\begin{eqnarray}
    \E[v(t)v(s)] &=& e^{-k(t-s)}\frac{\theta\sigma_v^2}{2k} + \theta^2, \\
    \cov (v(t),v(s)) &=& e^{-k(t-s)},
\end{eqnarray}
for $s<t$.
Furthermore, we have
\begin{eqnarray}
IV_{s,t} & = & \theta (t-s) - \frac{v(t)-v(s)}{k} + \frac{\sigma_v}{k}I_{s,t}, \\
IV_i & = & 
\theta (h-\tilde{h}) + \tilde{h} v_{i-1} - \frac{\sigma_v}{k}e^{-kih}\ie_i + \frac{\sigma_v}{k}I_i,
\label{eqn:IV_n}
\end{eqnarray}
where $\tilde{h} \triangleq (1-e^{-kh})/k$. The above formulas will be useful in calculating the moments and covariances of $y_i$ in the next section. 
Hereafter, we will use a generic subscript $n$ to denote variance and return (i.e., $v_n, y_n$) in their steady states.

\subsection{Moments and covariances of returns}

There are five parameters in the Heston model, therefore, to estimate them we need to obtain at least five moments/covariances of $y_n$.  In this section, we show specifically how $\E[y_n]$, $\var(y_n)$, $\cov(y_n,y_{n+1})$, $\cov(y_n,y_{n+2})$, and $\cov(y_n^2,y_{n+1})$ can be calculated and derive explicit formulas for them in terms of the five parameters involved in the Heston model, which can then be used to estimate these parameters.  

Based on (\ref{eq_yn}), together with $\E[I_n] = 0$, $\E[I^*_n] = 0$, and $\E[IV_n] = \theta h$, we have
\begin{equation}
    \E[y_n] = (\mu-\theta/2)h.
\end{equation}
To calculate $\var(y_n)$, we first note $\cov(IV_n,I_n^*) = 0$ 
, $\cov(v_{n-1},I_n^*) = 0$, $\cov(\ie_n,I_n^*) = 0$, and $\cov(I_n,I_n^*) = 0$.  Hence,
\begin{equation*}
\var(y_n) = \frac{1}{4}\var(IV_n) - \rho \cov(IV_n,I_n) + \rho^2 \var(I_n) + (1-\rho^2) \var(I_n^*).
\end{equation*}
Let us take $\var(IV_n)$ as an example, it can be further expanded as
\begin{align*}
    \var(IV_n)
    &= \tilde{h}^2\var(v_{n-1}) + \frac{\sigma_v^2}{k^2}\left( e^{-2knh}\var(\ie_n)  + \var(I_n) - 2e^{-knh}\cov(\ie_n,I_n)\right),
\end{align*}
and we have
\begin{equation*}
    \var (IV_n)
    = \frac{\theta \sigma_v^2}{k^2}(h-\tilde{h}).
\end{equation*}
Terms
$\cov(IV_n,I_n),\var(I_n)$ and $\var(I_n^*)$ can be computed similarly. In summary, we have
\begin{equation}
    \var(y_n) = \theta h + \left(\frac{\sigma_v^2}{4k^2}-\frac{\rho \sigma_v}{k} \right)\theta (h-\tilde{h}).
\label{eq.varyn}
\end{equation}

For $\cov(y_n,y_{n+1})$, $\cov(y_n,y_{n+2})$, we have
\begin{eqnarray}
\cov(y_n,y_{n+1}) & = & \theta\tilde{h}^2\left(\frac{\sigma_v^2}{8k} - \frac{\rho\sigma_v}{2}\right), \label{eq_ylag1} \\
\cov(y_n,y_{n+2}) & = & e^{-kh}\cov(y_n,y_{n+1}) =
e^{-kh} \theta\tilde{h}^2\left(\frac{\sigma_v^2}{8k} - \frac{\rho\sigma_v}{2}\right). \label{eq_ylag2}
\end{eqnarray} 
In fact, (\ref{eq_ylag2}) can be generalized as
\begin{equation}\label{eqn:cov-y-with-m-lagged-y}
    \cov(y_n,y_{n+m}) = e^{-(m-1)kh}\cov(y_n,y_{n+1}), 
\end{equation}
for $m = 2, 3,\cdots$.

Finally, we consider $\cov(y_n^2,y_{n+1})$.  Since the computation is fairly straightforward but tedious, we omit all the details and just provide the following final result:
\begin{align}
\cov(y_n^2,y_{n+1}) 
    &= \frac{\theta \sigma_v^4}{8k^3}\tilde{h}(he^{-kh} - \tilde{h}) + \left( \frac{\theta \sigma_v^2}{4k}\mu h - \frac{\theta^2\sigma_v^2}{8k}h - \frac{\theta\sigma_v^2}{4k} \right)\tilde{h}^2\nonumber\\
    &\quad - \frac{\rho\sigma_v}{2}\tilde{h}\bigg[  \left( \frac{3\sigma_v^2}{2k^2} - \frac{2\rho\sigma_v}{k}  \right)\theta(he^{-kh} -\tilde{h}) + (2\mu\theta -\theta^2)h\tilde{h} \bigg].
\end{align}
All the detailed calculations are provided in Appendix A.

Before closing, we point out that we can in fact compute the moments and covariances of \textit{any} (positive integer) order and a  recursive calculation procedure is provided in Appendix B.
And the calculation of the moments and covariances for other affine SV models is presented in Appendix D.

\section{Parameter estimation}
In this section, we derive our MM estimators for the five parameters in the Heston model based on asset price samples and the moments/covariances obtained in the previous section.  We prove both the moments/covariances estimators and the MM estimators satisfy the central limit theorem and also derive the explicit formulas for their covariance matrices.   

Assume we are given a sequence of asset price samples, $S_0,S_1,\cdots, S_N$, based on which we can calculate return samples:
\begin{equation*}
    Y_i \triangleq \log S_i - \log S_{i-1}, 
\end{equation*}
$i = 1,\cdots, N$.
These samples can be used to compute the sample moments and covariances of $y_n$, which can be used to estimate their population counterparts:
\begin{align*}
    \E[y_n] &\approx \overline{Y} \triangleq \frac{1}{N}\sum_{i=1}^N Y_i , \\
    \var(y_n)  &\approx S^2 \triangleq \frac{1}{N}\sum_{i=1}^N (Y_i - \overline{Y})^2 ,\\
    \cov(y_n,y_{n+1}) &\approx \hat{\cov}(y_n,y_{n+1}) \triangleq  \frac{1}{N-1}\sum_{i=1}^{N-1} (Y_i - \overline{Y})(Y_{i+1} - \overline{Y}) ,\\
    \cov(y_n,y_{n+2}) &\approx \hat{\cov}(y_n,y_{n+2}) \triangleq \frac{1}{N-2}\sum_{i=1}^{N-2} (Y_i - \overline{Y})(Y_{i+2} - \overline{Y}), \\
    \cov(y_n^2,y_{n+1}) &\approx \hat{\cov}(y_n^2,y_{n+1}) \triangleq \frac{1}{N-1}\sum_{i=1}^{N-1} (Y_i^2 - \overline{Y^2})(Y_{i+1} - \overline{Y}) ,
\end{align*}
where $\overline{Y^2} \triangleq (1/N)\sum_{i=1}^N Y_i^2$. 
Let 
\begin{align*}
    \boldsymbol{\gamma} &= (\E[y_n],\var(y_n),\cov(y_n,y_{n+1}),\cov(y_n,y_{n+2}), \cov(y_n^2,y_{n+1}))^T, \\
    \hat{\boldsymbol{\gamma}} &= (\overline{Y}, S^2, \hat{\cov}(y_n,y_{n+1}), \hat{\cov}(y_n,y_{n+2}), \hat{\cov}(y_n^2,y_{n+1}))^T.
\end{align*}
where $T$ denotes the transpose of a vector or matrix.  We also define
\begin{align*}
        z_i^1 &\triangleq (y_i-\E[y_n])^2, 
        &z_i^2 &\triangleq (y_i-\E[y_n])(y_{i+1}-\E[y_n]), \\
        z_i^3 &\triangleq (y_i-\E[y_n])(y_{i+2}-\E[y_n]),
        & z_i^4 &\triangleq (y_i^2-\E[y_n^2])(y_{i+1}-\E[y_n]), \\
        z_i &\triangleq (z_i^1,z_i^2,z_i^3,z_i^4). 
        &  &  
\end{align*}
Let  
\begin{align*}
        \sigma_{11} &= \lim_{N\rightarrow\infty}\frac{1}{N}\boldsymbol{1}^T_N[\cov(y_i,y_j)]_{N\times N} \boldsymbol{1}_N,\\
        \sigma_{1m} &= \lim_{N\rightarrow\infty}\frac{1}{N}\boldsymbol{1}^T_N[\cov(y_i,z_j^{m-1})]_{N\times N_m} \boldsymbol{1}_{N_m}, ~~~~~~~\mbox{$2\leq m\leq 5$}, \\
        \sigma_{lm} &= \lim_{N\rightarrow\infty}\frac{1}{N}\boldsymbol{1}^T_{N_k}[\cov(z_i^{l-1},z_j^{m-1})]_{N_l\times N_m} \boldsymbol{1}_{N_c}, ~~~2\leq l\leq m\leq 5,
\end{align*}
    where $(N_2,N_3,N_4,N_5) = (N,N-1,N-2,N-1)$ and $\boldsymbol{1}_{n} = (1,\cdots,1)^T$ with $n$ elements.  
We are now ready to present the following central limit theorem for our moment/covariance estimators.
\begin{theorem}
\label{thm4.1}
    \begin{equation}\label{eqn:CLT}
        \lim_{N\rightarrow\infty}\sqrt{N}(\hat{\boldsymbol{\gamma}} - \boldsymbol{\gamma}) \overset{d}{=} \mathcal{N}(0,\Sigma)
    \end{equation}
    where $\overset{d}{=}$ denotes equal in distribution and $\mathcal{N}(0,\Sigma)$ is a multivariate normal distribution with mean 0 and symmetric covariance matrix
    $\Sigma = [\sigma_{lm}]_{5\times 5}$, with
    diagonal entries
    \begin{align}
        \sigma_{11}
        & = \theta h + \theta h \left(\frac{\sigma_v^2}{4k^2} - \frac{\rho\sigma_v}{k} \right),\label{eqn:sigma11}\\
        \sigma_{22}
        &= \var(z_n^1) 
        + 2C_2\cov(z_n^1,v_n)
        + 2C \cov(z_n^1,v_n^2),\label{eqn:sigma22}\\
        \sigma_{33}
        &= \var(z_n^2) + 2 \cov(z_n^2,z_{n+1}^2) + 2C_3\cov(z_n^2,v_{n+1})\nonumber\\
        &\quad + 2e^{-kh}C \cov(z_n^2, v_{n+1}^2),\label{eqn:sigma33}\\
        \sigma_{44}
        &= \var(z_n^3) + 2\cov(z_n^3,z_{n+1}^3) + 2 \cov(z_n^3,z_{n+2}^3) + 2C_4 \cov(z_n^3,v_{n+2})\nonumber\\
        &\quad  + 2e^{-2kh}C \cov(z_n^3,v_{n+2}^2),\label{eqn:sigma44}\\
        \sigma_{55}
        &= \var(z_n^4) + 2\cov(z_n^4,z_{n+1}^4) + 2C_{51}\cov(z_n^4,v_{n+1}) + 2C_{52}\cov(z_n^4,v_{n+1}^2)\nonumber\\
        &\quad + 2C_{53}\cov(z_n^4,v_{n+1}^3),\label{eqn:sigma55}
    \end{align}
    and off-diagonal entries
\begingroup
\allowdisplaybreaks
    \begin{align*}
        \sigma_{12}
        &= \cov(y_n,z_n^1) + C_y \cov(v_n,z_n^1) + C_2 \cov(y_n,v_n) + C \cov(y_n, v_n^2),\\
        \sigma_{13}
        &= \cov(y_n,z_n^2) + \cov(y_{n+1},z_n^2) + C_y \cov(v_{n+1},z_n^2) + C_3 \cov(y_n,v_n)\nonumber\\
        &\quad  + e^{-kh}C \cov(y_n,v_n^2),\\
        \sigma_{14}
        &= \cov(y_n,z_n^3) + \cov(y_{n+1},z_n^3) + \cov(y_{n+2},z_n^3)+ C_y \cov(v_{n+2},z_n^3) \nonumber\\
        &\quad + C_4 \cov(y_n,v_n) + e^{-2kh}C \cov(y_n, v_n^2),\\
        \sigma_{15}
        &= \cov(y_n,z_n^4) + \cov(y_{n+1},z_n^4) + C_y \cov(v_{n+1},z_n^4)\nonumber\\
        &\quad + C_{51}\cov(y_n,v_{n}) + C_{52}\cov(y_n,v_{n}^2) + C_{53}\cov(y_n,v_{n}^3),
        \\
        \sigma_{23}
        &= \cov(z_n^1,z_n^2) + C_3 \cov(z_n^1, v_n) +  e^{-kh}C \cov(z_n^1, v_n^2) \nonumber \\
        &\quad + \cov(z_{n+1}^1,z_n^2) + C_2 \cov(v_{n+1}, z_n^2) + C \cov(v_{n+1}^2,z_n^2),\\
        \sigma_{24}
        &= \cov(z_n^1,z_n^3) + C_4 \cov(z_n^1, v_n) + e^{-2kh}C \cov(z_n^1, v_n^2) + \cov(z_{n+1}^1,z_n^3)\nonumber\\ 
        &\quad  + \cov(z_{n+2}^1,z_n^3) + C_2 \cov(v_{n+2}, z_n^3) + C \cov(v_{n+2}^2, z_n^3),\\
        \sigma_{25}
        &=\cov(z_n^1,z_n^4) + C_{51}\cov(z_n^1,v_n) + C_{52}\cov(z_n^1,v_n^2) + C_{53}\cov(z_n^1,v_n^3)\nonumber\\
        &\quad + \cov(z_{n+1}^1,z_n^4) + C_2 \cov(v_{n+1},z_n^4) + C_3 \cov(v_{n+1}^2,z_n^4),
    \\
        \sigma_{34}
        &= \cov(z_n^2,z_n^3) + \cov(z_n^2,z_{n+1}^3) + C_4 \cov(z_n^2,v_{n+1}) \nonumber\\
        &\quad  + e^{-2kh}C \cov(z_n^2,v_{n+1}^2) +  \cov(z_{n+1}^2,z_n^3) + \cov(z_{n+2}^2,z_n^3) \nonumber\\
        &\quad + C_3 \cov(v_{n+2},z_n^3) + e^{-kh}C \cov(v_{n+2}^2, z_n^3),\\
        \sigma_{35}
        &= \cov(z_n^2,z_n^4) + \cov(z_n^2,z_{n+1}^4)\nonumber\\
        &\quad + C_{51}\cov(z_n^2,v_{n+1}) +  C_{52}\cov(z_n^2,v_{n+1}^2) +  C_{53}\cov(z_n^2,v_{n+1}^3)\nonumber\\
        &\quad + \cov(z_{n+1}^2,z_n^4) + C_3\cov(v_{n+1},z_n^4) + e^{-kh}C \cov(v_{n+1}^2,z_n^4),\\
        \sigma_{45}
        &= \cov(z_n^3,z_n^4) + \cov(z_n^3,z_{n+1}^4) + \cov(z_n^3,z_{n+2}^4)\nonumber\\
        &\quad + C_{51}\cov(z_n^2,v_{n+2}) +  C_{52}\cov(z_n^2,v_{n+2}^2) +  C_{53}\cov(z_n^2,v_{n+2}^3)\nonumber\\
        &\quad + \cov(z_{n+1}^3,z_n^4) +  C_4 \cov(v_{n+1},z_n^4) + e^{-2kh}C \cov(v_{n+1}^2, z_n^4), 
    \end{align*}
\endgroup
    where $C_2,C_3,C_4,C,C_y,C_{51},C_{52},C_{53}$ are provided in Appendix C.
\end{theorem}

\begin{proof}
The variance process $v(t)$ is a Markov process, so is the discrete-time process $\{v_n\}$. Conditioning on $v_{n-1}$, $IV_n$ and $I_n$ are independent of $\{v_i, 0\le i\le n-2\}$. 
Hence, $\{(v_n,IV_n,I_n)\}$ is also a Markov process. 
The return process $\{ y_n\}
$ can be viewed as a stochastic process driven by the hidden Markov chain $\{(v_n,IV_n,I_n)\}$. Conditioning on $(v_n, IV_n, I_n)$, $y_n$'s are independent and normally distributed with mean $\mu h - IV_n/2 + \rho I_n$ and variance $(1-\rho^2)IV_n$, i.e.,
\begin{equation*}
    y_n|(v_n, IV_n, I_n) \sim \mathcal{N}\left(\mu h - \frac{1}{2}IV_n + \rho I_n, (1-\rho^2)IV_n\right).
\end{equation*}
Under the condition $\sigma_v^2 \le 2k\theta$, the strictly stationary variance process $v(t)$ (i.e., Cox-Ingersoll-Ross process \citep{cox1985theory}) is  $\rho$-mixing; see Proposition 2.8 in \citet{genon2000stochastic}. Hence, $v(t)$ is ergodic and strong mixing (also known as $\alpha$-mixing) with an exponential rate; see  \citet{bradley2005basic}.  The discrete-time process $\{v_n\}$ is also ergodic and inherits the exponentially fast strong mixing properties of $v(t)$. It can be verified that the hidden chain $\{(v_n,IV_n,I_n)\}$ is also ergodic and strong mixing with an exponential rate, noting that $I_n = (v_n - v_{n-1} - k\theta h + k IV_n)/\sigma_v$.
Furthermore, according to Proposition 3.1 in \citet{genon2000stochastic}, the observable process $\{y_n\}$ is strictly stationary and strong mixing with an exponential rate. 
Applying the central limit theorem for strictly stationary strong mixing sequences (see Theorem 1.7 in \cite{ibragimov1962some}), we can prove \eqref{eqn:CLT} as Corollary 3.1 in \citet{genon2000stochastic}.

Next, we provide the calculations for the diagonal entries of $\Sigma$, but omit those for the off-diagonal entries since they are very similar.  We now calculate the five diagonal entries of $\Sigma$.
\begingroup
\allowdisplaybreaks
\begin{itemize}
\item[$\sigma_{11}$:]  Based on \eqref{eqn:cov-y-with-m-lagged-y}, we have
\begin{eqnarray*}
\sigma_{11}
& = & \lim_{N\to\infty}\frac{1}{N}\left( N \var(y_n) + 2\sum_{1\leq i<j\leq N} \cov(y_i,y_j)\right) \\
& = & \var(y_n) + 2 \lim_{N\to\infty}\frac{1}{N} \sum_{1\leq i<j\leq N} e^{-(j-i-1)kh}\cov(y_n,y_{n+1}) \\
& = & \var(y_n) + 2 \lim_{N\to\infty}\frac{1}{N} \left[ \frac{N-1}{1-e^{-kh}} - \frac{e^{-kh}-e^{-Nkh}}{(1-e^{-kh})^2} \right]\cov(y_n,y_{n+1}) \\
& = & \var(y_n) + \frac{2}{1-e^{-kh}}\cov(y_n,y_{n+1}) \\
& = & \theta h + \theta h \left(\frac{\sigma_v^2}{4k^2} - \frac{\rho\sigma_v}{k} \right).
\end{eqnarray*}
\item[$\sigma_{22}$:] We have
\begin{eqnarray*}
\sigma_{22}
& = & \lim_{N\to\infty}\frac{1}{N}\left( N \var(z_n^1) + 2 \sum_{1\leq i<j\leq N} \cov(z_i^1,z_j^1)\right) \\
& = & \var(z_n^1)  + \frac{\tilde{h}^2}{2(1-e^{-2kh})}\cov(z_n^1,v_n^2)\\
&   & + 2\left[\frac{\tilde{h} - he^{-kh}}{1-e^{-kh}}\left(\frac{\sigma_v^2}{2k^2} - \frac{\rho\sigma_v}{k}\right) + \frac{1}{k} - \frac{\tilde{h}^2(\sigma_v^2/k + 2\theta)}{4(1-e^{-2kh})} \right]\cov(z_n^1,v_n),
\end{eqnarray*}
where the off-diagonal entries of $[\cov(z_i^1,z_j^1)]_{N\times N}$ can be calculated as:
\begin{align*}
   \cov(z_i^1,z_{j}^1)
   &= \frac{\tilde{h}^2}{4} e^{2k(j-i-1)h}\cov(z_n^1,v_n^2) - \frac{\tilde{h}^2}{4}e^{2k(j-i-1)h}\left(\frac{\sigma_v^2}{k} + 2\theta\right) \cov(z_n^1,v_n) \\
   &\quad + e^{-k(j-i-1)h}\left[\left(\frac{\sigma_v^2}{2k^2} - \frac{\rho\sigma_v}{k}\right)(\tilde{h} - he^{-kh}) + \tilde{h} \right]\cov(z_n^1,v_n).
\end{align*}
\item[$\sigma_{33}$:] We have
\begin{eqnarray*}
\sigma_{33}
& = & \lim_{N\to\infty}\frac{1}{N}\left( (N-1)\var(z_n^2)  + 2 \sum_{1\leq i<j\leq N-1} \cov(z_i^2,z_j^2)\right) \\ 
& = & \var(z_n^2) + 2\cov(z_n^2,z_{n+1}^2) + 2\lim_{N\to\infty}\frac{1}{N}\sum_{1< i+1 < j \leq N-1} \cov(z_i^2,z_j^2)\\
& = & \var(z_n^2) + 2 \cov(z_n^2,z_{n+1}^2) + \frac{\tilde{h}^2e^{-kh}}{2(1-e^{-2kh})}\cov(z_n^2,v_{n+1}^2)\\
&   & + 2\left[\frac{e^{-kh}h\tilde{h}}{1-e^{-kh}} \left(\frac{\sigma_v^2}{4k} - \frac{\rho\sigma_v}{2}\right) - \frac{e^{-kh}\tilde{h}^2}{1-e^{-2kh}}\left(\frac{\sigma_v^2}{4k} + \frac{\theta}{2} \right) \right]\cov(z_n^2,v_{n+1}),
\end{eqnarray*}
where $\cov(z_i^2,z_j^2)$ ($i+1<j$) can be calculated as
\begin{align*}
    \cov(z_i^2,z_{j}^2)
    &= \frac{\tilde{h}^2}{4} e^{-k(2j-2i-3)h}\cov(z_n^2,v_{n+1}^2)\\
    &\quad - \frac{\tilde{h}^2}{4} e^{-k(2j-2i-3)h}\left(\frac{\sigma_v^2}{k} + 2\theta\right) \cov(z_n^2,v_{n+1})\\
    &\quad + e^{-k(j-i-2)h}e^{-kh} \left(\frac{\sigma_v^2}{4k} - \frac{\rho\sigma_v}{2}\right)h\tilde{h}\cov(z_n^2,v_{n+1}).
\end{align*}
\item[$\sigma_{44}$:] We have
\begin{eqnarray*}
    \sigma_{44}
    & = & \lim_{N\to\infty}\frac{1}{N}\left( (N-2)\var(z_n^3) + 2\sum_{1\leq i<j \leq N-2} cov(z_i^3,z_j^3) \right) \\
    & = & \var(z_n^3) + 2\cov(z_n^3,z_{n+1}^3) + 2\cov(z_n^3,z_{n+2}^3) \\
    &   & + 2\lim_{N\to\infty}\frac{1}{N}\sum_{2<i+2<j\leq N-2} \cov(z_i^3,z_j^3)\\
    & = & \var(z_n^3) + 2\cov(z_n^3,z_{n+1}^3) + 2 \cov(z_n^3,z_{n+2}^3) + \frac{e^{-2kh}\tilde{h}^2}{2(1-e^{-2kh})}\cov(z_n^3,v_{n+2}^2) \\
    &   & + 2\left[ \frac{e^{-2kh}h\tilde{h}}{1-e^{-kh}}\left( \frac{\sigma_v^2}{4k} - \frac{\rho\sigma_v}{2} \right) - \frac{\tilde{h}^2 }{1-e^{-2kh}}\left(\frac{\sigma_v^2}{4k} + \frac{\theta}{2} \right) \right]\cov(z_n^3,v_{n+2}),
\end{eqnarray*}
where $\cov(z_i^3,z_j^3)$ ($ i+2 <j$) can be calculated as
\begin{align*}
    \cov(z_i^3,z_{j}^3)
    &= \frac{\tilde{h}^2}{4}e^{-2k(j-i-2)h}\cov(z_n^3,v_{n+2}^2) \\
    &\quad - e^{-2k(j-i-3)h}\left(\frac{\sigma_v^2}{4k} + \frac{\theta}{2} \right)\tilde{h}^2 \cov(z_n^3,v_{n+2})\\
    &\quad + e^{-k(j-i-1)h}\left( \frac{\sigma_v^2}{4k} - \frac{\rho\sigma_v}{2} \right)h\tilde{h}\cov(z_n^3,v_{n+2}).
\end{align*}
\item[$\sigma_{55}$:] We have
\begin{eqnarray*}
    \sigma_{55}
    & = & \lim_{N\to\infty}\frac{1}{N}\left( (N-1)\var(z_n^4) + 2\sum_{1\leq i<j \leq N-1} \cov(z_i^4,z_j^4) \right)\\
    & = & \var(z_n^4) + 2\cov(z_n^4,z_{n+1}^4) + 2\lim_{N\to\infty}\frac{1}{N}\sum_{1<i+1<j\leq N-1} \cov(z_i^4,z_j^4)\\
    & = & \var(z_n^4) + 2\cov(z_n^4,z_{n+1}^4) + 2C_{51}(z_n^4,v_{n+1}) + 2C_{52}(z_n^4,v_{n+1}^2) \\
    &   & + 2C_{53}\cov(z_n^4,v_{n+1}^3),
\end{eqnarray*}
where $\cov(z_i^4,z_j^4)$ ($i+1<j$) can be calculated as
\begin{align*}
    &\cov(z_i^4,z_{j}^4)\\
    &= - F_1 e^{-3k(j-i-2)h}\cov(z_n^4,v_{n+1}^3) \\
    &\quad + \left(F_4 e^{-3k(j-i-2)h} + F_5e^{-2k(j-i-2)h}\right)\cov(z_n^4,v_{n+1}^2)\\
    &\quad + \left(F_6e^{-3k(j-i-2)h}
    -F_7 e^{-2k(j-i-2)h}
    + F_8 e^{-k(j-i-2)h}\right)\cov(z_n^4,v_{n+1}).
\end{align*}
\end{itemize}
\endgroup
The off-diagonal entries of $\Sigma$ can be computed similarly. This completes our proof.
\end{proof}

\begin{remark}
Though the final explicit formulas for the entries of $\Sigma$ (except $\sigma_{11}$) are not given, they are not difficult to compute, as shown in Appendix B.  The complete derivations are quite tedious and the final formulas are lengthy.  
\end{remark}

Having obtained the moment and covariance estimates, we can now develop the estimators for the parameters.  First, let us consider $k$, the mean reversion velocity parameter in the variance (volatility) process.  Note that \eqref{eqn:cov-y-with-m-lagged-y} leads to
\begin{equation*}
    k = \frac{1}{(m-1)h}\ln \left( \frac{\cov(y_n,y_{n+1})}{\cov(y_n,y_{n+m})} \right), 
\end{equation*}
for $m = 2,3,\cdots$.  Therefore, we construct the following estimator for $k$:
\begin{equation}
    \hat{k} = \frac{1}{M-1}\sum_{m=2}^M\frac{1}{(m-1)h}\ln \left( \frac{\hat{\cov}(y_n,y_{n+1})}{\hat{\cov}(y_n,y_{n+m})} \right),
\end{equation}
where $2\leq  M < N$ (usually $M$ takes a small value, e.g., $M\leq 10$) and
\begin{align*}
    \hat{\cov}(y_n,y_{n+m}) \triangleq \frac{1}{N-m}\sum_{n=1}^{N-m}(Y_n - \overline{Y})(Y_{n+m} -\overline{Y}).
\end{align*}

The estimators for the other four parameters are given by
\begin{align}
    \hat{\theta} &= S^2/h - \frac{2(h-\tilde{h}_{\hat{k}})}{h\hat{k}\tilde{h}_{\hat{k}}^2}\hat{\cov}(y_n,y_{n+1}),\\
    \hat{\mu} &= \overline{Y}/h + \hat{\theta}/2,\\
    \hat{\sigma}_v^2 &= \frac{4\hat{k}\overline{Y} + [8\hat{d}_h/(\hat{\theta} \tilde{h}_{\hat{k}}^3)]\hat{\cov}(y_n,y_{n+1}) - 2\hat{k}\frac{\hat{\cov}(y_n^2,y_{n+1})}{\hat{\cov}(y_n,y_{n+1})} }{\hat{\theta}\tilde{h}_{\hat{k}}^2/(2\hat{\cov}(y_n,y_{n+1})) - \hat{d}_h/(\hat{k}\tilde{h}_{\hat{k}}) },\\
    \hat{\rho} &= \frac{\hat{\sigma}_{v}}{4\hat{k}} - \frac{2}{\hat{\theta}\hat{\sigma}_{v}\tilde{h}_{\hat{k}}^2}\hat{\cov}(y_n,y_{n+1}),
\end{align}
where $\tilde{h}_{\hat{k}} \triangleq (1-e^{-\hat{k}h})/\hat{k}$ and $\hat{d}_h \triangleq he^{-\hat{k}h} - \tilde{h}_{\hat{k}}$.

In what follows, we let $M=2$ in $\hat{k}$ and 
denote the corresponding estimators as a continuous mapping $g: \mathbb{R}^{5} \rightarrow \mathbb{R}^5$, i.e.,
\begin{align*}
    g(\hat{\boldsymbol{\gamma}}) = (\hat{k}, \hat{\theta},\hat{\sigma}_v,\hat{\mu},\hat{\rho})^T.
\end{align*}
(Our analysis remains the same for $M>2$.)
We now present the following central limit theorem for the parameter estimators.
\begin{theorem}
    Suppose $M=2$. Then,
    \begin{equation}
        \lim_{N\rightarrow\infty}\sqrt{N}((\hat{k}, \hat{\theta},\hat{\sigma}_v,\hat{\mu},\hat{\rho})^T - (k,\theta,\sigma_v,\mu,\rho)^T) \overset{d}{=} \mathcal{N}(0,J_g\Sigma J_g^{T})
    \end{equation}
    where $J_g$ is the Jacobian of $g$.
\end{theorem}
\begin{proof}
    It follows directly by applying the delta method (Theorem 3.1 in \citet{van2000asymptotic}) and Theorem 4.1.  $J_g$ is a $5\times 5$ matrix of partial derivatives with respect to the entries of $g$. For instance, the first row of $J_g$ is the gradient of $k$ with respect to $\boldsymbol{\gamma}$, i.e.,
    \begin{equation*}
        \left(0,0, \frac{1}{h\cov(y_n,y_{n+1})}, - \frac{1}{h\cov(y_n,y_{n+2})},0\right).
    \end{equation*}
\end{proof}

\section{Numerical experiments}
In this section, we provide numerical experiments to test the estimators derived in the previous sections.  First, we test the accuracy of our estimators under different parameter value settings.  Then, we validate the asymptotic behaviors of our estimators as the size of samples increases. Next, we vary sample interval $h$ to analyze its effect on our estimators. 
Finally, we compare our method with the MCMC method.

In the first set of experiments, we consider six different settings of parameter values, with S0 being the base setting in which $\mu=0.125$, $k=0.1$, $\theta=0.25$, $\sigma_v=0.1$, $\rho=-0.7$.  Each of the other five settings differs from S0 in only one parameter value: S1 increases $\mu$'s value to $0.4$, S2 decreases $k$'s value to $0.03$, S3 increases $\theta$'s value to $0.5$, S4 increases $\sigma_v$'s value to $0.2$, and S5 decreases (in absolute value) $\rho$'s value to $-0.3$. The values of the parameters in the volatility process, $k$, $\theta$, and $\sigma_v$, are the same as those in \citet{bollerslev2002estimating} except for S3. 

To generate sample observations for the Heston model, we use the first order Euler approximation method in which we set  $h=1$ and partition each interval into 20 smaller segments for approximation. For each parameter setting, we run 400 replications with 400K samples for each replication. The numerical results are presented in Table~\ref{tab:estimation-under-different-parSettings} with the format ``mean $\pm$ standard deviation" based on these 400 replications (the format remains the same for all numerical results in this section).  The results show our estimators are fairly accurate.  Of course, the estimation accuracy can be improved by increasing $N$ as will be shown by the next set of experiments.  

\begin{table}
    \centering
    \caption{Numerical results under different parameter settings}
    \label{tab:estimation-under-different-parSettings}
    \begin{tabular}{cccccc}
     \toprule
        Setting  & $\mu$ & $k$ & $\theta$ & $\sigma_v$ & $\rho$  \\
       \midrule
        S0 &  \emph{0.125}	&\emph{0.1}	&\emph{0.25}	&\emph{0.1}	&\emph{-0.7}\\
            &0.125±0.001	&0.101±0.015	&0.25±0.001	&0.1±0.009	&-0.706±0.043\\
            \cmidrule{2-6}
        S1 &\textbf{\emph{0.4}}	&\emph{0.1}	&\emph{0.25}	&\emph{0.1}	&\emph{-0.7}\\
            & 0.4±0.001	&0.1±0.015	&0.249±0.001	&0.1±0.009	&-0.71±0.047\\
            \cmidrule{2-6}
        S2 &\emph{0.125}	&\textbf{\emph{0.03}}	&\emph{0.25}	&\emph{0.1}	&\emph{-0.7}\\
             &0.125±0.001	&0.03±0.01	&0.25±0.003	&0.099±0.018	&-0.742±0.184\\
             \cmidrule{2-6}
        S3 &\emph{0.125}	&\emph{0.1}	&\textbf{\emph{0.5}}	&\emph{0.1}	&\emph{-0.7}\\
             &0.125±0.001	&0.099±0.013	&0.499±0.002	&0.1±0.009	&-0.711±0.055\\
             \cmidrule{2-6}
        S4 &\emph{0.125}	&\emph{0.1}	&\emph{0.25}	&\textbf{\emph{0.2}}	&\emph{-0.7}\\
             &0.125±0.001	&0.101±0.007	&0.249±0.002	&0.2±0.008	&-0.708±0.028\\
             \cmidrule{2-6}
        S5 &\emph{0.125}	&\emph{0.1}	&\emph{0.25}	&\emph{0.1}	&\textbf{\emph{-0.3}}\\
             &0.125±0.001	&0.103±0.026	&0.25±0.001	&0.101±0.015	&-0.304±0.034\\
       \bottomrule
    \end{tabular}
\end{table}

To test the effect of $N$ on the performance of our estimators, we run the second set of experiments in which we increase $N$ from $25K$ to $1600K$ for S0. The results are provided in Table~\ref{tab:changeN}, which show as $N$ increases the accuracy of our estimators improves at the rate of $1/\sqrt{N}$.

\begin{table}
    \centering
    \caption{Asymptotic behavior as $N$ increases}
    \label{tab:changeN}
    \begin{tabular}{cccccc}
      \toprule
       & $\mu$ & $k$ & $\theta$ & $\sigma_v$ & $\rho$  \\
       \cmidrule{2-6}
      N & $0.125$ & $0.1$ & $0.25$ & $0.1$ & $-0.7$ \\
      \midrule
        25K &0.124±0.004 &0.112±0.059 &0.249±0.004 &0.106±0.034 &-0.796±0.363\\
        100K &0.125±0.002 &0.102±0.03 &0.25±0.002 &0.1±0.019 &-0.726±0.105\\
        400K &0.125±0.001 &0.1±0.014 &0.25±0.001 &0.1±0.008 &-0.711±0.044\\
        1600K &0.125±0 &0.1±0.007 &0.25±0.001 &0.1±0.004 &-0.708±0.022\\
      \bottomrule
    \end{tabular}
\end{table}

We run the third set of experiments to test the effect of $h$ on our estimators.  Again, we consider the parameter setting S0 and increase $h$ from 0.5 to 2.  The results are given in Table~\ref{tab:changeh}.  It is clear that the accuracy of our estimators is not very sensitive to the value of $h$.  

\begin{table}
    \centering
    \caption{Effect of the value of $h$}
    \label{tab:changeh}
    \begin{tabular}{cccccc}
      \toprule
       & $\mu$ & $k$ & $\theta$ & $\sigma_v$ & $\rho$  \\
       \cmidrule{2-6}
      $h$ & $0.125$ & $0.1$ & $0.25$ & $0.1$ & $-0.7$ \\
      \midrule
        0.5 &0.125±0.001 &0.1±0.053 &0.25±0.001 &0.099±0.028 &-0.797±0.354\\
        1 &0.125±0.001 &0.101±0.015 &0.25±0.001 &0.1±0.009 &-0.712±0.046\\
        2 &0.125±0.001 &0.1±0.005 &0.25±0.001 &0.1±0.004 &-0.709±0.035\\
        4 &0.125±0.001 &0.101±0.004 &0.249±0.001 &0.1±0.003 &-0.711±0.029\\
      \bottomrule
    \end{tabular}
\end{table}

Finally, we compare our method with the MCMC method proposed in \cite{eraker2003impact} and \cite{johannes2010mcmc}; for the detailed implementation of the MCMC method, see \cite{Wu_MCMC_for_Estimating_2023}. The comparison is based on the parameter setting S0, again with 400 replications with 400K samples for each replication (the results are very similar for other parameter settings).  For the MCMC method, we use $\mu=0.0625,k=0.05,\theta=0.125,\sigma_v=0.05,\rho=-0.35$ as the initial values for the five parameters in the iteration procedure. For each replication, we run MCMC for 10,000 iterations with the first 5000 iterations as warm-up. The results are given in Table~\ref{tab:comparison-with-mcmc}, which shows our method performs much better than MCMC, both in terms of accuracy and computational time.  MCMC is very time-consuming and it takes more than one hour for each replication while our method takes less than a second. 

\begin{table}
    \centering
    \caption{Comparison with MCMC estimation}
    \label{tab:comparison-with-mcmc}
    \begin{tabular}{cccc}
      \toprule
         & True value & Our method & MCMC \\
      \midrule
        $\mu$ & 0.125 & 0.125±0.001 & 0.125±0.001\\
        $k$ & 0.1 &0.101±0.015 &0.063±0.001 \\
        $\theta$ & 0.25 &0.250±0.001 & 0.231±0.001\\
        $\sigma_v$ & 0.1 &0.100±0.009 &0.053±0.001 \\
        $\rho$ & -0.7 &-0.706±0.043 &-0.559±0.010 \\
        average time (seconds) & - & 0.29 & 4,280 \\
      \bottomrule
    \end{tabular}
\end{table}

\section{Conclusion}
In this paper, we consider the problem of parameter estimation for affine SV models.  We first obtain the moments and covariances of the asset price based on which we develop MM estimators.  A key of our method is to develop a recursive procedure to calculate the moments and covariances of the asset price.  Our MM estimators are simple and easy to implement.  We also establish the central limit theorem for our estimators in which the explicit formulas for the asymptotic covariance matrix are given.  Finally, we provide numerical results to test and validate our method.  The method proposed in this paper can potentially be applied to other SV models.  This is a possible future research direction.  

\section*{Acknowledgements}
This work was supported by the  National Natural Science Foundation of China (NSFC) under Grants 72033003 and 72350710219; the China Postdoctoral Science Foundation under Grant 2023M732054; the Shandong Provincial Natural Science Foundation under Grant ZR2023QG159; and the Shandong Postdoctoral Science Foundation under Grant SDCX-RS-202303004.

\section*{Data Availability Statement}
Data sharing is not applicable to this article as no new data were created or analyzed in this study.

\bibliography{references}

\begin{thebibliography}{}

\bibitem[A{\"\i}t-Sahalia and Kimmel, 2007]{ai2007maximum}
A{\"\i}t-Sahalia, Y. and Kimmel, R. (2007).
\newblock Maximum likelihood estimation of stochastic volatility models.
\newblock {\em Journal of financial economics}, 83(2):413--452.

\bibitem[Andersen et~al., 2015]{andersen2015parametric}
Andersen, T.~G., Fusari, N., and Todorov, V. (2015).
\newblock Parametric inference and dynamic state recovery from option panels.
\newblock {\em Econometrica}, 83(3):1081--1145.

\bibitem[Bansal et~al., 1994]{Bansal1994}
Bansal, R., Gallant, A.~R., Hussey, R., and Tauchen, G. (1994).
\newblock {\em Computational Aspects of Nonparametric Simulation Estimation},
  pages 3--22.
\newblock Springer Netherlands, Dordrecht.

\bibitem[Barndorff-Nielsen and Shephard, 2001]{barndorff2001non}
Barndorff-Nielsen, O.~E. and Shephard, N. (2001).
\newblock Non-gaussian ornstein--uhlenbeck-based models and some of their uses
  in financial economics.
\newblock {\em Journal of the Royal Statistical Society: Series B (Statistical
  Methodology)}, 63(2):167--241.

\bibitem[Bates, 1996]{bates1996jumps}
Bates, D.~S. (1996).
\newblock Jumps and stochastic volatility: Exchange rate processes implicit in
  deutsche mark options.
\newblock {\em The Review of Financial Studies}, 9(1):69--107.

\bibitem[Bollerslev and Zhou, 2002]{bollerslev2002estimating}
Bollerslev, T. and Zhou, H. (2002).
\newblock Estimating stochastic volatility diffusion using conditional moments
  of integrated volatility.
\newblock {\em Journal of Econometrics}, 109(1):33--65.

\bibitem[Bradley, 2005]{bradley2005basic}
Bradley, R.~C. (2005).
\newblock Basic properties of strong mixing conditions. a survey and some open
  questions.
\newblock {\em Probability surveys}, 2:107--144.

\bibitem[Chacko and Viceira, 2003]{chacko2003spectral}
Chacko, G. and Viceira, L.~M. (2003).
\newblock Spectral gmm estimation of continuous-time processes.
\newblock {\em Journal of Econometrics}, 116(1):259--292.
\newblock Frontiers of financial econometrics and financial engineering.

\bibitem[Cox et~al., 1985]{cox1985theory}
Cox, J.~C., Ingersoll~Jr, J.~E., and Ross, S.~A. (1985).
\newblock A theory of the term structure of interest rates.
\newblock {\em Econometrica}, 53:385--407.

\bibitem[Duffie et~al., 2000]{duffie2000transform}
Duffie, D., Pan, J., and Singleton, K. (2000).
\newblock Transform analysis and asset pricing for affine jump-diffusions.
\newblock {\em Econometrica}, 68(6):1343--1376.

\bibitem[Duffie and Singleton, 1993]{duffie1993simulated}
Duffie, D. and Singleton, K.~J. (1993).
\newblock Simulated moments estimation of markov models of asset prices.
\newblock {\em Econometrica}, 61(4):929--952.

\bibitem[Durham and Gallant, 2002]{durham2002numerical}
Durham, G. and Gallant, A. (2002).
\newblock Numerical techniques for maximum likelihood estimation of
  continuous-time diffusion processes.
\newblock {\em Journal of Business and Economic Statistics}, 20(3):297--338.

\bibitem[Eraker, 2001]{eraker2001mcmc}
Eraker, B. (2001).
\newblock Mcmc analysis of diffusion models with application to finance.
\newblock {\em Journal of Business \& Economic Statistics}, 19(2):177--191.

\bibitem[Eraker et~al., 2003]{eraker2003impact}
Eraker, B., Johannes, M., and Polson, N. (2003).
\newblock The impact of jumps in volatility and returns.
\newblock {\em The Journal of Finance}, 58(3):1269--1300.

\bibitem[Fama, 1965]{fama1965behavior}
Fama, E.~F. (1965).
\newblock The behavior of stock-market prices.
\newblock {\em The journal of Business}, 38(1):34--105.

\bibitem[Feunou and Okou, 2018]{Feunou2018risk}
Feunou, B. and Okou, C. (2018).
\newblock Risk-neutral moment-based estimation of affine option pricing models.
\newblock {\em Journal of Applied Econometrics}, 33(7):1007--1025.

\bibitem[Gallant and Tauchen, 1996]{gallant1996moments}
Gallant, A.~R. and Tauchen, G. (1996).
\newblock Which moments to match?
\newblock {\em Econometric theory}, 12(4):657--681.

\bibitem[Genon-Catalot et~al., 2000]{genon2000stochastic}
Genon-Catalot, V., Jeantheau, T., and Lar{\'e}do, C. (2000).
\newblock Stochastic volatility models as hidden markov models and statistical
  applications.
\newblock {\em Bernoulli}, 6(6):1051--1079.

\bibitem[Griffin and Steel, 2006]{griffin2006inference}
Griffin, J.~E. and Steel, M.~F. (2006).
\newblock Inference with non-gaussian ornstein--uhlenbeck processes for
  stochastic volatility.
\newblock {\em Journal of Econometrics}, 134(2):605--644.

\bibitem[Harvey and Shephard, 1996]{harvey1996estimation}
Harvey, A.~C. and Shephard, N. (1996).
\newblock Estimation of an asymmetric stochastic volatility model for asset
  returns.
\newblock {\em Journal of Business \& Economic Statistics}, 14(4):429--434.

\bibitem[Heston, 1993]{heston1993closed}
Heston, S.~L. (1993).
\newblock A closed-form solution for options with stochastic volatility with
  applications to bond and currency options.
\newblock {\em The review of financial studies}, 6(2):327--343.

\bibitem[Hurn et~al., 2015]{hurn2014estimating}
Hurn, A.~S., Lindsay, K.~A., and McClelland, A.~J. (2015).
\newblock Estimating the parameters of stochastic volatility models using
  option price data.
\newblock {\em Journal of Business \& Economic Statistics}, 33(4):579--594.

\bibitem[Ibragimov, 1962]{ibragimov1962some}
Ibragimov, I.~A. (1962).
\newblock Some limit theorems for stationary processes.
\newblock {\em Theory of Probability \& Its Applications}, 7(4):349--382.

\bibitem[Jacquier et~al., 2002]{jacquier2002bayesian}
Jacquier, E., Polson, N.~G., and Rossi, P.~E. (2002).
\newblock Bayesian analysis of stochastic volatility models.
\newblock {\em Journal of Business \& Economic Statistics}, 20(1):69--87.

\bibitem[Jasra et~al., 2011]{jasra2011inference}
Jasra, A., Stephens, D.~A., Doucet, A., and Tsagaris, T. (2011).
\newblock Inference for l{\'e}vy-driven stochastic volatility models via
  adaptive sequential monte carlo.
\newblock {\em Scandinavian Journal of Statistics}, 38(1):1--22.

\bibitem[Jiang and Knight, 2002]{jiang2002estimation}
Jiang, G.~J. and Knight, J.~L. (2002).
\newblock Estimation of continuous-time processes via the empirical
  characteristic function.
\newblock {\em Journal of Business \& Economic Statistics}, 20(2):198--212.

\bibitem[Johannes and Polson, 2010]{johannes2010mcmc}
Johannes, M. and Polson, N. (2010).
\newblock Mcmc methods for continuous-time financial econometrics.
\newblock In {\em Handbook of Financial Econometrics: Applications}, pages
  1--72. Elsevier.

\bibitem[Overbeck and Rydén, 1997]{overbeck1997}
Overbeck, L. and Rydén, T. (1997).
\newblock Estimation in the cox-ingersoll-ross model.
\newblock {\em Econometric Theory}, 13(3):430--461.

\bibitem[Peng et~al., 2014]{peng2014gradient}
Peng, Y., Fu, M.~C., and Hu, J.-Q. (2014).
\newblock Gradient-based simulated maximum likelihood estimation for
  lévy-driven ornstein–uhlenbeck stochastic volatility models.
\newblock {\em Quantitative Finance}, 14(8):1399--1414.

\bibitem[Peng et~al., 2016]{peng2016gradient}
Peng, Y., Fu, M.~C., and Hu, J.-Q. (2016).
\newblock Gradient-based simulated maximum likelihood estimation for stochastic
  volatility models using characteristic functions.
\newblock {\em Quantitative Finance}, 16(9):1393--1411.

\bibitem[Roberts et~al., 2004]{roberts2004bayesian}
Roberts, G.~O., Papaspiliopoulos, O., and Dellaportas, P. (2004).
\newblock Bayesian inference for non-gaussian ornstein–uhlenbeck stochastic
  volatility processes.
\newblock {\em Journal of the Royal Statistical Society: Series B (Statistical
  Methodology)}, 66(2):369--393.

\bibitem[Ross, 2010]{ross2010introduction}
Ross, S.~M. (2010).
\newblock {\em Introduction to probability models}.
\newblock Academic press.

\bibitem[Ruiz, 1994]{ruiz1994quasi}
Ruiz, E. (1994).
\newblock Quasi-maximum likelihood estimation of stochastic volatility models.
\newblock {\em Journal of econometrics}, 63(1):289--306.

\bibitem[Singleton, 2001]{singleton2001estimation}
Singleton, K.~J. (2001).
\newblock Estimation of affine asset pricing models using the empirical
  characteristic function.
\newblock {\em Journal of Econometrics}, 102(1):111--141.

\bibitem[Van~der Vaart, 2000]{van2000asymptotic}
Van~der Vaart, A.~W. (2000).
\newblock {\em Asymptotic statistics}, volume~3.
\newblock Cambridge university press.

\bibitem[Wu et~al., 2022]{wu2022moment}
Wu, Y., Hu, J., and Yang, X. (2022).
\newblock Moment estimators for parameters of lévy-driven ornstein–uhlenbeck
  processes.
\newblock {\em Journal of Time Series Analysis}, 43(4):610--639.

\bibitem[Wu et~al., 2019]{wu2019moment}
Wu, Y., Hu, J., and Zhang, X. (2019).
\newblock Moment estimators for the parameters of ornstein-uhlenbeck processes
  driven by compound poisson processes.
\newblock {\em Discrete Event Dynamic Systems}, 29(1):57--77.

\bibitem[Wu, 2023]{Wu_MCMC_for_Estimating_2023}
Wu, Y.-F. (2023).
\newblock {hestonmcmc: an R package for Estimating Heston Stochastic Volatility
  Models through MCMC}.
\newblock https://github.com/xmlongan/hestonmcmc.git.

\bibitem[Wu and Hu, 2023]{Wu_ajdmom_a_python_2023}
Wu, Y.-F. and Hu, J.-Q. (2023).
\newblock {ajdmom: a Python Package for Deriving Moment Formulas of Affine Jump
  Diffusion Processes}.
\newblock https://github.com/xmlongan/ajdmom.git.

\bibitem[Yang et~al., 2021]{yang2020method}
Yang, X., Wu, Y., Zheng, Z., and Hu, J.-Q. (2021).
\newblock Method of moments estimation for l\'{e}vy-driven ornstein–uhlenbeck
  stochastic volatility models.
\newblock {\em Probability in the Engineering and Informational Sciences},
  35(4):975--1004.

\end{thebibliography}

\appendix
\section*{Appendix A: Covariance computation}
Here we provide some intermediate results omitted in Section 2.

To calculate $\var(y_n)$, we first note $\cov(IV_n,I_n^*) = 0$, $\cov(v_{n-1},I_n^*) = 0$, $\cov(\ie_n,I_n^*) = 0$, and $\cov(I_n,I_n^*) = 0$.  Hence,
\begin{equation*}
\var(y_n) = \frac{1}{4}\var(IV_n) - \rho \cov(IV_n,I_n) + \rho^2 \var(I_n) + (1-\rho^2) \var(I_n^*).
\end{equation*}
First, we note $\var(I_n^*) = \E[(I_n^*)^2] = \E[IV_n] = \theta h.$
We also have
\begin{align*}
    \var(IV_n)
    &= \tilde{h}^2\var(v_{n-1}) + \frac{\sigma_v^2}{k^2}\left( e^{-2knh}\var(\ie_n)  + \var(I_n) - 2e^{-knh}\cov(\ie_n,I_n)\right).
\end{align*}
The four components in the above equation are given by
\begin{align*}
    \var(v_{n-1}) &= \frac{\theta\sigma_v^2}{2k},
    &\var(\ie_n) &= \frac{\theta [e^{2knh}-e^{2k(n-1)h}]}{2k},\\
    \var(I_n) &= \theta h,
    &\cov(I_n, \ie_n) &= \frac{\theta [e^{knh}-e^{k(n-1)h}]}{k}.
\end{align*}
In what follows, we illustrate how $\cov(I_n, \ie_n)$ can be computed in detail.  Since
\begin{eqnarray*}
    dI_{n-1,t} & = & \sqrt{v(t)}dw^v(t), \\
    d\ie_{n-1,t} & = & e^{kt}\sqrt{v(t)}dw^v(t),
\end{eqnarray*}
we have by It\^{o} formula
\begin{equation*}
    d(I_{n-1,t}\ie_{n-1,t}) = \left(\ie_{n-1,t}\sqrt{v(t)} + I_{n-1,t}e^{-kt}\sqrt{v(t)}\right)dw^v(t) + e^{kt}v(t)dt.
\end{equation*}
Note
\begin{equation*}
    \int_{(n-1)h}^{t}\left(\ie_{n-1,s}\sqrt{v(s)} + I_{n-1,s}e^{-ks}\sqrt{v(s)}\right)dw^v(s)
\end{equation*}
is a martingale with zero mean, so
\begin{equation*}
    \E[I_{n-1,t}\ie_{n-1,t}]
    = \int_{(n-1)h}^{t}e^{ks}\E[v(s)]ds = \frac{\theta [e^{kt}-e^{k(n-1)h}]}{k}. 
\end{equation*}
The interchange of expectation and integration above is guaranteed by Lebesgue's dominated convergence theorem. By letting $t=nh$, we can obtain 
$$
\cov(I_n,\ie_n)= \E[I_n\ie_n] - \E[I_n]E[\ie_n] = \E[I_n\ie_n]
= \frac{\theta [e^{knh}-e^{k(n-1)h}]}{k}.
$$
In summary, we have 
\begin{equation*}
    \var (IV_n)
    = \frac{\theta \sigma_v^2}{k^2}(h-\tilde{h}).
\end{equation*}

For $\cov(y_n,y_{n+1})$, $\cov(y_n,y_{n+2})$, we have
\begin{eqnarray*}
   \cov(y_n,y_{n+1}) &=&\frac{1}{4}\cov(IV_n,IV_{n+1}) - \frac{1}{2}\rho \cov(I_{n},IV_{n+1}), \\
   \cov(y_n,y_{n+2}) &=&\frac{1}{4}\cov(IV_n,IV_{n+2}) - \frac{1}{2}\rho \cov(I_{n},IV_{n+2}),
\end{eqnarray*}
since $\cov(IV_n,I_{n+i}^*)=\cov(I_n^*,IV_{n+i})=\cov(I_n,I_{n+i}^*)=0$ for $i\geq 0$ and $\cov(I_n^*,I_{n+i}^*)=0$ for $i\geq 1$.
After some calculations, we can obtain
\begin{eqnarray*}
\cov(IV_n,IV_{n+1}) &=& \tilde{h} \cov(IV_n, v_{n}) = \frac{\theta \tilde{h}^2\sigma_v^2}{2k}, \\
\cov(I_n,IV_{n+1}) &=& \tilde{h}^2\theta\sigma_v, \\
\cov(IV_n,IV_{n+2}) &=& e^{-kh}\cov(IV_n,IV_{n+1}),\\ 
\cov(I_n,IV_{n+2})
    &=& e^{-kh}\cov(I_n,IV_{n+1}),
\end{eqnarray*}
hence, we have (\ref{eq_ylag1}) and (\ref{eq_ylag2}).

Finally, we consider $\cov(y_n^2,y_{n+1})$.  First, we have
\begin{eqnarray*}
\cov(y_n^2,y_{n+1})
    &=& -\frac{1}{8}\tilde{h}\cov(IV_n^2,v_n) - \frac{1}{2}\tilde{h}(1-\rho^2)\cov((I_n^*)^{2}, v_n) \\
    & & + \frac{1}{2} \mu h \tilde{h} \cov(IV_n, v_n)
    - \frac{1}{2}\tilde{h}\rho^2\cov(I_n^2,v_n) \\
    & & - \mu h \tilde{h}\rho \cov(I_n,v_n) + \frac{1}{2}\tilde{h}\rho \cov(IV_nI_n,v_n).
\end{eqnarray*}
We have already computed $\cov(IV_n,v_n)$ earlier.  The other covariances can be computed based on the following variances and covariances:
\begin{align*}
    &\var(v_{n-1}),& &\var(\ie_n),& &\cov(I_n,\ie_n),\\
    &\cov(v_{n-1}^2,v_{n-1}),& &\cov(v_{n-1}\ie_n,v_{n-1}),& &\cov(v_{n-1}I_n,v_{n-1}),\\
    &\cov(v_{n-1}\ie_n,\ie_n),& &\cov(v_{n-1}I_n,\ie_n),& &\cov(\ie_n^2,v_{n-1}),\\
    &\cov(\ie_n^2,\ie_n),& &\cov(\ie_nI_n,v_{n-1}),& &\cov(\ie_nI_n,\ie_n),\\
    &\cov(I_n^2,v_{n-1}),& &\cov(I_n^2,\ie_n),& &\cov((I_n^*)^2,v_{n-1}).
\end{align*}
In what follows, we use $\cov(\ie_nI_n,\ie_n)$ as an illustrative example to show how these terms can be computed.
Using It\^{o} formula, we have
$$
    d(\ie_{n-1,t}^2I_{n-1,t})
    = \alpha (t)dw^v(t)
     + \beta (t) dt,
$$
where
\begin{eqnarray*}
\alpha (t)
&=& \ie_{n-1,t}^2\sqrt{v(t)} + 2I_{n-1,t}\ie_{n-1,t}e^{kt}\sqrt{v(t)},\\
\beta (t)
&=& I_{n-1,t}e^{2kt}v(t) + 2\ie_{n-1,t}e^{kt}v(t).
\end{eqnarray*}
Therefore, $\ie_{n-1,t}^2I_{n-1,t}$ can be expressed as
\begin{equation*}
    \ie_{n-1,t}^2I_{n-1,t} = \int_{(n-1)h}^t \alpha(s)dw^v(s) + \int_{(n-1)h}^t \beta(s)ds.
\end{equation*}
The expectation of the first integral on the right hand side of the above equation is zero, hence
\begin{equation*}
    \E[\ie_{n-1,t}^2I_{n-1,t}] = \int_{(n-1)h}^t\E[\beta(s)]ds.
\end{equation*}
(We have again interchanged the expectation and integration operations in the above equation by applying Lebesgue's dominated convergence theorem.)  We can easily compute $\E[I_{n-1,t}v(t)]$ and $\E[\ie_{n-1,t}v(t)]$ in $\E[\beta(t)]$ by expanding $v(t)$ via \eqref{eqn:v(t)}, which gives us
\begin{align*}
    \E[\ie_{n-1,t}^2I_{n-1,t}]
    &= \frac{\theta \sigma_v}{k}\left[ \frac{1}{k}e^{2kt} - \frac{1}{k}e^{kt+k(n-1)h} - e^{2k(n-1)h}[t-(n-1)h] \right].
\end{align*}
By letting $t=nh$ and also noting $E[\ie_n]=0$, we have
\begin{equation*}
\cov(\ie_nI_n,\ie_n) = \E[\ie_n^2I_n]
    = \frac{\theta \sigma_v}{k^2}e^{2knh}(1-e^{-kh}-khe^{-2kh}).
\end{equation*}
All other variances and covariances can be computed in a similar way. 
Moreover, we point out that we can in fact compute the moments and covariances of \textit{any} (positive integer) order and a  recursive calculation procedure is provided in Appendix B.

\section*{Appendix B: A recursive procedure for computing moments and covariances}
Here we discuss how the moments and covariances of $y_n$ can be calculated.  Define 
\begin{align*}
    y_{n-1,t} 
    &\triangleq \mu [t-(n-1)h] - \frac{1}{2}IV_{n-1,t} + \rho I_{n-1,t} + \sqrt{1-\rho^2}I_{n-1,t}^*,
\end{align*}
then
\begin{align*}
    &y_{n-1,t} - \E[y_{n-1,t}]\\
    &= \beta_{n-1,t}\theta - \beta_{n-1,t}v_{n-1}
     + \frac{\sigma_v}{2k}e^{-kt}\ie_{n-1,t} + \left(\rho - \frac{\sigma_v }{2k}\right)I_{n-1,t} + \sqrt{1-\rho^2} I_{n-1,t}^*,
\end{align*}
where $\beta_{n-1,t} \triangleq (1-e^{-k[t-(n-1)h]})/(2k)$.
The $l$th central moment of $y_{n-1,t}$, denoted by $\cm_l(y_{n-1,t})$, can be computed based on the following quantities:
\begin{equation}\label{eqn:comb_moment}
    \E[\theta^{n_1}v_{n-1}^{n_2}(e^{-kt}\ie_{n-1,t})^{n_3}I_{n-1,t}^{n_4}I_{n-1,t}^{*n_{5}}], 
\end{equation}
where $n_i\geq 0$ ($i=1,2,3,4,5$) and $\sum_{i=1}^{5}n_i=l$.  We can compute \eqref{eqn:comb_moment} as in the following two steps:
\begin{equation*}
 \E[\theta^{n_1}v_{n-1}^{n_2}\E[(e^{-kt}\ie_{n-1,t})^{n_3}I_{n-1,t}^{n_4}I_{n-1,t}^{*n_{5}}|v_{n-1}]],
\end{equation*}
i.e., first take expectation conditioning on $v_{n-1}$, and then take expectation w.r.t.\ $v_{n-1}$.  We will show later that the conditional moment $\E[\ie_{n-1,t}^{n_3}I_{n-1,t}^{n_4}I_{n-1,t}^{*n_{5}}|v_{n-1}]$ is a polynomial of $v_{n-1}$, which implies that \eqref{eqn:comb_moment} can be expressed as the moments of $v_{n-1}$.  
By using \eqref{eqn:v_moment}, we can compute the moment of any order of $v_{n-1}$, hence we can compute \eqref{eqn:comb_moment} as well.  

Next, we consider $\E[\ie_{n-1,t}^{n_3}I_{n-1,t}^{n_4}I_{n-1,t}^{*n_{5}}|v_{n-1}]$.  We separate $\ie_{n-1,t}^{n_3}I_{n-1,t}^{n_4}I_{n-1,t}^{*n_{5}}$ into two parts: $\ie_{n-1,t}^{n_3}I_{n-1,t}^{n_4}$ and $I_{n-1,t}^{*n_{5}}$, since they are driven by two different Wiener processes $w^v(t)$ and $w(t)$, respectively. For $\ie_{n-1,t}^{n_3}I_{n-1,t}^{n_4}$, we have
\begin{align*}
    d(\ie_{n-1,t}^{n_3}I_{n-1,t}^{n_4})
    &= c_w(t) dw^v(t)+ c(t) dt
\end{align*}
where
\begin{align*}
    c_w(t) 
    &\triangleq n_3 \ie_{n-1,t}^{n_3-1}I_{n-1,t}^{n_4}\sqrt{v(t)} + n_4 \ie_{n-1,t}^{n_3}I_{n-1,t}^{n_4-1}e^{kt}\sqrt{v(t)},\\
    c(t)
    &\triangleq \bigg[\frac{1}{2}n_3(n_3-1)\ie_{n-1,t}^{n_3-2}I_{n-1,t}^{n_4}e^{2kt} + \frac{1}{2}n_4(n_4-1)\ie_{n-1,t}^{n_3}I_{n-1,t}^{n_4-2}\\
    &\qquad + n_3n_4\ie_{n-1,t}^{n_3-1}I_{n-1,t}^{n_4-1}e^{kt} \bigg] v(t).
\end{align*}
The conditional expectation
\begin{equation*}
    \E[\ie_{n-1,t}^{n_3}I_{n-1,t}^{n_4}|v_{n-1}] = \int_{(n-1)h}^t \E[c(s)|v_{n-1}]ds.
\end{equation*}
Since
\begin{align*}
    v(t)
    &= e^{k(n-1)h}(v_{n-1}-\theta) e^{-kt} + \theta + \sigma_v e^{-kt}\ie_{n-1,t},
\end{align*}
$\E[\ie_{n-1,t}^{n_3}I_{n-1,t}^{n_4}|v_{n-1}]$ can be expressed as
\begin{align}
\E[\ie_{n-1,t}^{n_3}I_{n-1,t}^{n_4}|v_{n-1}] & 
= f_1(\E[\ie_{n-1,t}^{n_3\boldsymbol{-2}}I_{n-1,t}^{n_4}|v_{n-1}],
    \E[\ie_{n-1,t}^{n_3}I_{n-1,t}^{n_4\boldsymbol{-2}}|v_{n-1}],\nonumber\\
    &\qquad~ \E[\ie_{n-1,t}^{n_3\boldsymbol{-1}}I_{n-1,t}^{n_4\boldsymbol{-1}}|v_{n-1}],
     \E[\ie_{n-1,t}^{n_3\boldsymbol{-1}}I_{n-1,t}^{n_4}|v_{n-1}],\nonumber\\
    &\qquad~ \E[\ie_{n-1,t}^{n_3\boldsymbol{+1}}I_{n-1,t}^{n_4\boldsymbol{-2}}|v_{n-1}],
     \E[\ie_{n-1,t}^{n_3}I_{n-1,t}^{n_4\boldsymbol{-1}}|v_{n-1}])\label{eqn:recursive-2items},
\end{align}
where $f_1$ is some function whose details have been left out for simplicity.  

For $I_{n-1,t}^{*n_5}$, we have
\begin{align*}
    dI_{n-1,t}^{*n_5}
    &= n_5I_{n-1,t}^{*n_5-1}\sqrt{v(t)} dw(t) + \frac{1}{2}n_5(n_5-1)I_{n-1,t}^{*n_5-2}v(t)dt.
\end{align*}
Note that $d(\ie_{n-1,t}^{n_3}I_{n-1,t}^{n_4})dI_{n-1,t}^{*n_5} = 0$ because $dw^v(t)dw(t) = 0$.
Hence, 
\begin{align*}
& d(\ie_{n-1,t}^{n_3}I_{n-1,t}^{n_4}I_{n-1,t}^{*n_5}) \\
    &= (\ie_{n-1,t}^{n_3}I_{n-1,t}^{n_4})dI_{n-1,t}^{*n_5} + I_{n-1,t}^{*n_5}d(\ie_{n-1,t}^{n_3}I_{n-1,t}^{n_4})\\
    &= n_5\ie_{n-1,t}^{n_3}I_{n-1,t}^{n_4}I_{n-1,t}^{*n_5-1}\sqrt{v(t)} dw^s(t) + c_w(t)I_{n-1,t}^{*n_5}dw^v(t)\\
    &\quad + \left[\frac{1}{2}n_5(n_5-1) \ie_{n-1,t}^{n_3}I_{n-1,t}^{n_4}I_{n-1,t}^{*n_5-2}v(t)+ c(t)I_{n-1,t}^{*n_5}\right]dt.
\end{align*}
Therefore,
\begin{align*}
    &\E[\ie_{n-1,t}^{n_3}I_{n-1,t}^{n_4}I_{n-1,t}^{*n_5}|v_{n-1}]\\
    & = \int_{(n-1)h}^t\E\left[\frac{1}{2}n_5(n_5-1) \ie_{n-1,s}^{n_3}I_{n-1,s}^{n_4}I_{n-1,s}^{*n_5-2}v(s)+ c(s)I_{n-1,s}^{*n_5}|v_{n-1}\right]ds.
\end{align*}
Again, there exists some function $f_2$ based on which $\E[\ie_{n-1,t}^{n_3}I_{n-1,t}^{n_4}I_{n-1,t}^{*n_5}|v_{n-1}]$ can be expressed as
\begin{align}
    &\E[\ie_{n-1,t}^{n_3}I_{n-1,t}^{n_4}I_{n-1,t}^{*n_5}|v_{n-1}]\nonumber\\
    &= f_2(\E[\ie_{n-1,t}^{n_3}I_{n-1,t}^{n_4}I_{n-1,t}^{*n_5\boldsymbol{-2}}|v_{n-1}],
    \E[\ie_{n-1,t}^{n_3\boldsymbol{+1}}I_{n-1,t}^{n_4}I_{n-1,t}^{*n_5\boldsymbol{-2}}|v_{n-1}],\nonumber\\
    &\qquad~ \E[\ie_{n-1,t}^{n_3\boldsymbol{-2}}I_{n-1,t}^{n_4}I_{n-1,t}^{*n_5}|v_{n-1}],
    \E[\ie_{n-1,t}^{n_3}I_{n-1,t}^{n_4\boldsymbol{-2}}I_{n-1,t}^{*n_5}|v_{n-1}],\nonumber\\
    &\qquad~ \E[\ie_{n-1,t}^{n_3\boldsymbol{-1}}I_{n-1,t}^{n_4\boldsymbol{-1}}I_{n-1,t}^{*n_5}|v_{n-1}],
    \E[\ie_{n-1,t}^{n_3\boldsymbol{-1}}I_{n-1,t}^{n_4}I_{n-1,t}^{*n_5}|v_{n-1}],\nonumber\\
    &\qquad~ \E[\ie_{n-1,t}^{n_3\boldsymbol{+1}}I_{n-1,t}^{n_4\boldsymbol{-2}}I_{n-1,t}^{*n_5}|v_{n-1}], \E[\ie_{n-1,t}^{n_3}I_{n-1,t}^{n_4\boldsymbol{-1}}I_{n-1,t}^{*n_5}|v_{n-1}]), \label{eqn:recursive-3items}. 
\end{align}
It should be noted that 
\begin{equation*}
    \E[I_{n-1,t}^{*n_5}|v_{n-1}] = \E[I_{n-1,t}^{n_5}|v_{n-1}].
\end{equation*}
\eqref{eqn:recursive-2items} and \eqref{eqn:recursive-3items} can be used to compute the central moment of any order of $y_{n-1,t}$ recursively, from lower order ones to high order ones.  For example, we can start with the combinations $\{(n_3,n_4,n_5), l=1\}$, then $ \{(n_3,n_4,n_5), l=2\}$, and so on, where $n_3+n_4+n_5=l$. The computations are fairly straightforward but computationally intensive, which can be automated, see \cite{Wu_ajdmom_a_python_2023} for a Python implementation.

\section*{Appendix C: Definitions of $C_2$,$C_3$,$C_4$,$C$,$C_y$,$C_{51}$,$C_{52}$,$C_{53}$}
Here we define $C_2,C_3,C_4,C,C_y,C_{51},C_{52},C_{53}$, which are used in Theorem \ref{thm4.1}. 
\begingroup
\allowdisplaybreaks
\begin{align*}
        D_1
        &\triangleq -\mu h\tilde{h} - \frac{\rho\sigma_v}{k}(\tilde{h} - he^{-kh}) + \frac{1}{2}\theta\tilde{h}(h-\tilde{h}) + \frac{\sigma_v^2}{4k^2}(\tilde{h} + \tilde{h}e^{-kh} - 2he^{-kh}) + \tilde{h}, \\
        D_2
        &\triangleq -\mu h \tilde{h}e^{-kh} - \frac{\rho\sigma_v}{k}e^{-kh}(h -2he^{-kh} + \tilde{h}) + \frac{\theta}{4}\tilde{h}(\tilde{h} - 3\tilde{h}e^{-kh} + 2he^{-kh})\\
        &\quad  + \frac{\sigma_v^2}{4k^2}e^{-kh}(3\tilde{h}e^{-kh} - 4he^{-kh} - \tilde{h} + 2h) + \tilde{h}e^{-kh}, \\ 
        D_3
        &\triangleq (\mu h)^2 e^{-kh} + \mu h\theta(2\tilde{h}e^{-kh} -\tilde{h} - he^{-kh}) -\frac{\sigma_v^2}{k}\mu h (h-\tilde{h})e^{-kh}\\
        &\quad  + 2\mu h^2\rho\sigma_ve^{-kh} + \frac{\rho^2\sigma_v^2}{k}( kh^2+ h - \tilde{h}) - \frac{\rho^2\sigma_v}{k}( h - \tilde{h})\\
        &\quad - \frac{\rho\theta\sigma_v}{k}(\tilde{h} -3\tilde{h}e^{-kh} + 2he^{-2kh} + kh^2e^{-kh} - kh\tilde{h}e^{-kh} ) \\
        &\quad - \frac{\rho\sigma_v^3}{k^2} (2he^{-kh}-2\tilde{h} + kh^2)e^{-kh} + \theta(\tilde{h} + he^{-kh} - 2\tilde{h}e^{-kh})\\ 
        &\quad  + \frac{\theta^2}{4}(h-\tilde{h})(2\tilde{h} + e^{-kh}h - 3e^{-kh}\tilde{h}) + \frac{\sigma_v}{k}(h -\tilde{h})\\
        &\quad + \frac{\theta\sigma_v^2}{8k^2}(-4\tilde{h} + 5\tilde{h}e^{-kh} - 5\tilde{h}e^{-2kh} + 4he^{-2kh}) \\
        &\quad + \frac{\sigma_v^4}{8k^3}e^{-kh}(2kh^2-3\tilde{h} -2h -3\tilde{h}e^{-kh} + 8he^{-kh} ), \\
        D_{33}
        &\triangleq 3\left[ \left(2\theta^2 + \frac{3\theta\sigma_v^2}{k} + \frac{\sigma_v^4}{k^2} \right)\frac{1+2e^{kh}}{1+e^{-kh}} + \left(\theta^2 + \frac{\theta\sigma_v^2}{k} + \frac{\sigma_v^4}{2k^2} \right)\frac{1-e^{-kh}}{1+e^{-kh}} \right],\\
        D_{32}
        &\triangleq 3\left( 2\theta^2 + \frac{3\theta\sigma_v^2}{k} + \frac{\sigma_v^4}{k^2} \right),\\
        D_{31}
        &\triangleq 3\left[ \theta^2 + \frac{2\theta\sigma_v^2}{k} + \frac{\sigma_v^4}{2k^2} - \frac{\theta\sigma_v^2}{k(1+e^{-kh})} \right].
\end{align*}
\begin{align*}
& F_1 \triangleq \frac{1}{8}\tilde{h}^3e^{-kh}, 
& F_2 & \triangleq \frac{1}{8}\theta \tilde{h}^3 - \frac{1}{2}D_2\tilde{h}, \\
& F_3 \triangleq \frac{1}{2}D_1\theta\tilde{h} - \frac{1}{2}D_3\tilde{h} + \frac{1}{2}\tilde{h}e^{-kh} E[y^2],
& F_4 & \triangleq 3F_1(\theta + \sigma_v^2/k), \\
& F_5 \triangleq F_2 -3F_1 (\theta + \sigma_v^2/k),
& F_6 & \triangleq F_1D_{33}, \\
& F_7 \triangleq F_1 D_{32}  + F_2 (2\theta + \sigma_v^2/k ),
& F_8 & \triangleq F_3  -F_1D_{31} + F_2 (2\theta + \sigma_v^2/k).
\end{align*}
\begin{align*}
         C_{2}
         &\triangleq \frac{\tilde{h} - he^{-kh}}{1-e^{-kh}}\left(\frac{\sigma_v^2}{2k^2} - \frac{\rho\sigma_v}{k}\right) - \frac{\tilde{h}^2}{1-e^{-2kh}} \left(\frac{\sigma_v^2}{4k} + \frac{\theta}{2} \right) + \frac{1}{k} ,\\
         C_3 
         &\triangleq \frac{e^{-kh}h\tilde{h}}{1-e^{-kh}} \left(\frac{\sigma_v^2}{4k} - \frac{\rho\sigma_v}{2}\right) - \frac{e^{-kh}\tilde{h}^2}{1-e^{-2kh}}\left(\frac{\sigma_v^2}{4k} + \frac{\theta}{2} \right),\\
         C_4
        &\triangleq  \frac{e^{-2kh}h\tilde{h}}{1-e^{-kh}}\left( \frac{\sigma_v^2}{4k} - \frac{\rho\sigma_v}{2} \right) - \frac{\tilde{h}^2 }{1-e^{-2kh}}\left(\frac{\sigma_v^2}{4k} + \frac{\theta}{2} \right) ,\\
         C
         &\triangleq \frac{\tilde{h}^2}{4(1-e^{-2kh})}, \\
         C_y 
         & \triangleq - \frac{\tilde{h}}{2(1-e^{-kh})},\\
        C_{51}
        &\triangleq \frac{F_6}{1-e^{-3kh}} - \frac{F_7}{1-e^{-2kh}} + \frac{F_8}{1-e^{-kh}},\\
        C_{52}
        &\triangleq \frac{F_4}{1-e^{-3kh}} + \frac{F_5}{1-e^{-2kh}}, \\
        C_{53}
        & \triangleq -\frac{F_1}{1-e^{-3kh}}.
\end{align*}
\endgroup

\section*{Appendix D: Extensions to other affine stochastic volatility models}
Here we discuss three extensions of the baseline affine SV model studied in the main body of this paper.  
\subsection*{Appendix D1: SV model with jumps in the return process}
The first extension is the following SV model, which adds a jump component in the log price process of the Heston model: 
\begin{align*}
    d\log s(t) &= (\mu- v(t)/2) dt + \sqrt{v(t)}dw^s(t) + dz^s(t),\\
    dv(t)      &= k(\theta - v(t))dt + \sigma_v \sqrt{v(t)}dw^v(t),
\end{align*}
where $z^s(t)$ is a compound Poisson process and independent of everything else,  $N(t)$ is the associated Poisson process with rate $\lambda$, and $j$ is the jump which is a random variable distributed according to $F_j(\cdot,\boldsymbol{\theta}_j)$ and $\boldsymbol{\theta}_j$ is a parameter. 
For this model,
\begin{equation*}
    y_n = y_{o,n} + J_n,
\end{equation*}
where
\begin{eqnarray*}
    y_{o,n} & \triangleq & \mu h - \frac{1}{2}IV_{n} + \rho I_n + \sqrt{1-\rho^2}I_n^{*}, \\
    J_n & \triangleq & \sum_{i=N((n-1)h)+1}^{N(nh)}j_i.
\end{eqnarray*}
In this model, we have two more parameters $\boldsymbol{\theta}_j$ and $\lambda$ to estimate. Therefore, we need the following seven equations: 
\begin{eqnarray*}
\E [y_n] & = & E [y_{o,n}] + E [J_n], \\
\var(y_n) &= & \var(y_{o,n}) + \var(J_n),\\
\cov(y_n,y_{n+1}) &= & \cov(y_{o,n},y_{o,n+1}),\\
\cov(y_n,y_{n+2}) &= & e^{-kh}\cov(y_n,y_{n+1}),\\
\cov(y_n^2,y_{n+1}) &= & \cov(y_{o,n}^2,y_{o,n+1}) + 2 E [J_n]\cov(y_{o,n},y_{o,n+1}),\\
\cov(y_n,y_{n+1}^2) &=& \cov(y_{o,n},y_{o,n+1}^2) + 2 E [J_{n+1}]\cov(y_{o,n},y_{o,n+1}), \\
\cm_3[y_n] &= & \cm_3[y_{o,n}] + \cm_3[J_n],
\end{eqnarray*}
where $\cm_3[\cdot]$ denotes the third central moment. If $\boldsymbol{\theta}_j$ contains more than one parameter, then more moment equations can be computed.

\subsection*{Appendix D2: SV model with jumps in the variance process}
The second extension adds a jump component in the variance process of the Heston model:
\begin{align*}
d\log s(t) &= (\mu - v(t)/2) dt + \sqrt{v(t)}dw^s(t),\\
dv(t) &= k(\theta - v(t))dt + \sigma_v\sqrt{v(t)}dw^v(t) + dz^v(t),
\end{align*}
where $z^v(t)$ denotes a compound Poisson process with rate $\lambda$ and is independent of all other processes. The jump sizes are determined by a random variable, $j$, with distribution $F_j(\cdot,\boldsymbol{\theta}_j)$, where $\boldsymbol{\theta}_j$ symbolizes the parameter vector. As in the Heston model, initial variance $v(0)$ is assumed to follow the steady state distribution of $v(t)$ to ensure the process is strictly stationary.

Applying It\^{o}'s lemma, we derive:
\begin{equation*}
d(e^{kt}v(t)) = \theta de^{kt} + \sigma_v e^{kt}\sqrt{v(t)}dw^v(t) + e^{kt}dz(t).
\end{equation*}
This leads to:
\begin{align*}
v(t) - \theta = e^{-kt}(v(0) - \theta)  + \sigma_v e^{-kt} \int_0^t e^{ks}\sqrt{v(s)}dw^v(s) + e^{-kt} \int_0^t e^{ks}dz(s).
\end{align*}
Given the stationarity of $v(t)$, we have $v(0)\overset{d}{=} v(t)$. Consequently, the distributional equality holds:
\begin{equation*}
(1-e^{-kt}) (v(t)-\theta) \overset{d}{=} \sigma_v e^{-kt} \int_0^te^{ks}\sqrt{v(s)}dw^v(s) + e^{-kt}\int_0^te^{ks}dz^v(s).
\end{equation*}

Defining the integrated exponential jump terms as $I\!E\!Z_{0,t} \triangleq \int_0^t e^{ks}dz^v(s)$, we can rewrite the equation as:
\begin{equation*}
(1-e^{-kt})(v(t)-\theta) \overset{d}{=} \sigma_v e^{-kt} I\!E_{0,t} + e^{-kt} I\!E\!Z_{0,t}.
\end{equation*}
For an interval $((n-1)h, nh]$, this relationship becomes:
\begin{equation*}
 (1-e^{-knh}) (v_n - \theta) \overset{d}{=} \sigma_v e^{-knh} I\!E_n + e^{-knh} I\!E\!Z_n,
\end{equation*}
where $I\!E\!Z_n \triangleq I\!E\!Z_{(n-1)h, nh}$.

The $m$-th moment of $(v_n-\theta)$ can be calculated using:
\begin{equation}\label{eqn:moment-svj-in-variance}
 \mathbb{E}[(v_n - \theta)^m] 
 = (1-e^{-knh})^{-m} e^{-mknh} \sum_{i=0}^m \binom{m}{i}\sigma_v^i \mathbb{E}[I\!E_n^i] \mathbb{E}[I\!E\!Z_n^{m-i}],~~ m = 1,2,\cdots,
\end{equation}
due to the independence between $I\!E_n$ and $I\!E\!Z_n$. The expectation $\mathbb{E}[I\!E_n^i]$ can be computed as detailed in Appendix A. In what follows we demonstrate how to compute the expectation $\mathbb{E}[I\!E\!Z_n^m]$. 

By definition, we have
\begin{equation*}
I\!E\!Z_n = \sum_{i = N((n-1)h) + 1}^{N(nh)} e^{ks_i}j_i.
\end{equation*}
To illustrate, consider the expected value $\E[I\!E\!Z_1]$ for the interval $(0,h]$. Assume that the compound Poisson process $z(t)$ experiences jumps at times $s_1,s_2,\cdots, s_{N(h)}$, corresponding to jump sizes $j_1,j_2,\cdots, j_{N(h)}$. It is important to note that $s_1,s_2,\cdots, s_{N(h)}$ are not ordered. Given $N(h)$, these times are independent and identically distributed (i.i.d.) uniform random variables over the interval $(0, h]$; this is detailed in Theorem 5.2 by \citet{ross2010introduction}. 

Representing a generic uniform random variable over $(0, h]$ with $s$, one can interpret $I\!E\!Z_{0,t}$ as a variant of the process $z^v(t)$ with identical Poisson dynamics but with modified jump sizes $e^{ks}j$ instead of $j$. 

The $m$-th moment of $\E[I\!E\!Z_1]$ is calculated via:
\begin{equation}\label{eqn:compound-moment}
\E[I\!E\!Z_1^m]
= \sum_{n=1}^{\infty} e^{-\lambda h}\frac{(\lambda h)^n}{n!} \E\left[ \left( \sum_{i=1}^n e^{ks_i}j_i \right)^m \right].
\end{equation}
Computing the first moment is straightforward: $\E[I\!E\!Z_1] = \lambda h \E[e^{ks}]\E[j]$. The calculation of the second and higher moments is feasible provided the moments of $j$ exist, as demonstrated in studies by \citet{wu2019moment}, \citet{wu2022moment} and \citet{yang2020method}. Consequently, the moments of $v_n-\theta$ of any (integer) order are determinable through Equations~\eqref{eqn:moment-svj-in-variance} and \eqref{eqn:compound-moment}.

Employing Equation~\eqref{eqn:moment-svj-in-variance}, the $m$-th moment of $v_n$ can be calculated as follows:
\begin{equation*}
\E[v_n^m] = \E[[(v_n-\theta) + \theta]^m] = \sum_{i=1}^m\binom{m}{i}\E[(v_n-\theta)^i]\theta^{m-i}.
\end{equation*}
Likewise, the $n$-th interval return $y_n$ can be decomposed as:
\begin{equation*}
y_n = \mu h - \frac{1}{2}IV_n + \rho I_n + \sqrt{1-\rho^2}I_n^*.
\end{equation*}
The moments and covariances of the returns are computable using the methodology described in Appendix B. We omit the detailed intermediary steps here due to their complexity.

\subsection*{Appendix D3: Multi-factor SV model}
We consider the SV model in which the volatility consists of two independent component processes,
\begin{align*}
    d\log s(t) &= (\mu- v(t)/2) dt + \sqrt{v(t)}dw(t),\\
    v(t)       &= v_1(t) + v_2(t),\\
    dv_1(t)    &= k_1(\theta_1 - v_1(t))dt + \sigma_{1v} \sqrt{v_1(t)}dw_1(t),\\
    dv_2(t)    &= k_2(\theta_2 - v_2(t))dt + \sigma_{2v} \sqrt{v_2(t)}dw_2(t),
\end{align*}
where $w_1(t)$, $w_2(t)$ are two independent Wiener processes, which are also independent of $w(t)$. 

We define
\begin{align*}
    I_{1,n} &\triangleq \int_{(n-1)h}^{nh}\sqrt{v_1(t)}dw_1(t),
    &I_{2,n} &\triangleq \int_{(n-1)h}^{nh}\sqrt{v_2(t)}dw_2(t),\\
    IV_{1,n}&\triangleq \int_{(n-1)h}^{nh}v_1(t)dt,
    &IV_{2,n}&\triangleq \int_{(n-1)h}^{nh}v_2(t)dt,\\
    IV_{n} &\triangleq IV_{1,n} + IV_{2,n},
    &I_n^* &\triangleq \int_{(n-1)h}^{nh}\sqrt{v(t)}dw(t).
\end{align*}
The return $y_n$ can be expressed as
\begin{equation*}
    y_n = \mu h - \frac{1}{2}(IV_{1,n} + IV_{2,n}) + I_n^*.
\end{equation*}
There are seven parameters to estimate in this model, so we need the following seven equations
\begingroup
\allowdisplaybreaks
\begin{align*}
    \E[y_n] &= \mu h -\frac{1}{2}(\theta_1 + \theta_2)h,\\
    \var(y_n) &= \frac{1}{4}\var(IV_{1,n}) + \frac{1}{4}\var(IV_{2,n}) + \var(I_n^*),\\
    \cov(y_n,y_{n+1}) &= \frac{1}{4}\cov(IV_{1,n},IV_{1,n+1}) + \frac{1}{4}\cov(IV_{2,n},IV_{2,n+1}),\\
    cov(y_n,y_{n+2})
    &= e^{-k_1h}\frac{1}{4}cov(IV_{1,n}, IV_{1,n+1})
      +e^{-k_2h}\frac{1}{4}cov(IV_{2,n}, IV_{2,n+1}),\\
    \cm_3[y_n] &= -\frac{1}{8}\cm_3[IV_{1,n}] -\frac{1}{8}\cm_3[IV_{2,n}],\\
    \cov(y_n^2,y_{n+1})&= -\frac{1}{8}\cov(IV_{1,n}^2,IV_{1,n+1}) - \frac{1}{4}(2-2\mu h + \theta_2h)\cov(IV_{1,n},IV_{1,n+1})\\
    &\quad -\frac{1}{8}\cov(IV_{2,n}^2,IV_{2,n+1}) - \frac{1}{4}(2-2\mu h + \theta_1h)\cov(IV_{2,n},IV_{2,n+1}),\\
    \cov(y_n,y_{n+1}^2)&= -\frac{1}{8}\cov(IV_{1,n},IV_{1,n+1}^2) -\frac{1}{4}(2-2\mu h +\theta_2 h) \cov(IV_{1,n},IV_{1,n+1})\\
    &\quad -\frac{1}{8}\cov(IV_{2,n},IV_{2,n+1}^2) -\frac{1}{4}(2-2\mu h + \theta_1 h) \cov(IV_{2,n},IV_{2,n+1}).\\
\end{align*}
\endgroup
\eqref{eqn:recursive-2items} can be used to compute the intermediate terms $\var(IV_{i,n})$, $\cov(IV_{i,n},IV_{i,n+1})$, $cm_3[IV_{i,n}]$, $\cov(IV_{i,n}^2,IV_{i,n+1})$, $\cov(IV_{i,n},IV_{i,n+1}^2)$
($i=1,2$). The final formulas are lengthy so we omit them here.

\end{document}